\documentclass[sigconf, nonacm, screen, authorversion]{acmart}

\usepackage{booktabs} 

\usepackage[utf8]{inputenc}
\usepackage[T2A,T1]{fontenc}
\usepackage[russian,english]{babel}

\usepackage{amsmath,amsfonts}
\usepackage{algorithmic}
\usepackage{csquotes}
\usepackage{graphicx}
\usepackage{textcomp}
\usepackage{xcolor}
\usepackage{amsthm}
\usepackage{mathtools}
\usepackage{enumitem}
\usepackage{thm-restate}
\usepackage{caption}
\usepackage{tikz}
\usetikzlibrary {graphs,shapes.geometric}

\makeatletter
\def\old@comma{,}
\catcode`\,=13
\def,{\ifmmode \old@comma\discretionary{}{}{}\else \old@comma \fi }
\makeatother

\setcopyright{none}

\begin{document}

\hyphenation{e-xists}
\hyphenation{mis-sing}
\hyphenation{pa-ra-do-xi-cal}

\newcommand{\vm}[1]{\mbox{}{\color{brown}\textbf{\small [VM: #1]}}}
\newcommand{\ek}[1]{\mbox{}{\color{blue}\textbf{\small [EK: #1]}}}
\newcommand{\of}[1]{\mbox{}{\color{teal}\textbf{\small [OF: #1]}}}

\newcommand{\cD}{\mathcal{D}}\newcommand{\cE}{\mathcal{E}}\newcommand{\cP}{\mathcal{P}}\newcommand{\cF}{\mathcal{F}}
\newcommand{\cQ}{\mathcal{Q}}\newcommand{\cO}{\mathcal{O}}\newcommand{\cI}{\mathcal{I}}\newcommand{\cC}{\mathcal{C}}
\newcommand{\cR}{\mathcal{R}}
\newcommand{\cU}{\mathcal{U}}
\newcommand{\cS}{\mathcal{S}}

\newcommand{\cT}{\mathcal{T}}\newcommand{\bC}{\mathbf{C}}\newcommand{\cK}{\mathcal{K}}\newcommand{\cG}{\mathcal{G}}\newcommand{\cL}{\mathcal{L}}\newcommand{\bbP}{\mathbb{P}}\newcommand{\fA}{\mathfrak{A}}\newcommand{\fB}{\mathfrak{B}}\newcommand{\fC}{\mathfrak{C}}\newcommand{\fD}{\mathfrak{D}}\newcommand{\fG}{\mathfrak{G}}

\renewcommand{\phi}{\varphi} \newcommand{\eps}{\varepsilon}

\newcommand{\AAA}{\mbox{\large \boldmath $\BBB_{*}$}}
\newcommand{\AAAp}{\mbox{\large \boldmath $\BBB'_{*}$}}
\newcommand{\BBB}{\mbox{\large \boldmath $\beta$}}

\newcommand{\EQ}{\ensuremath{{\mathcal EQ}}}
\newcommand{\Sat}{\ensuremath{\textit{Sat}}}
\newcommand{\FinSat}{\ensuremath{\textit{FinSat}}}

\newcommand{\FO}{\mbox{\rm FO}}
\newcommand{\FOt}{\mbox{$\mbox{\rm FO}^2$}}
\newcommand{\Ct}{\mbox{$\mathcal{C}^2$}}
\newcommand{\UOF}{\mbox{$\mbox{\rm UF}_1$}}
\newcommand{\ODF}{\mbox{$\mbox{\rm UF}_1$}}
\newcommand{\SUOF}{\mbox{$\mbox{\rm SUF}_1$}}
\newcommand{\RUOF}{\mbox{$\mbox{\rm RUF}_1$}}

\newcommand{\GFt}{\mbox{$\mbox{\rm GF}^2$}}
\newcommand{\GF}{\mbox{$\mbox{\rm GF}$}}
\newcommand{\GFTG}{\mbox{$\mbox{\rm GF+TG}$}}
\newcommand{\FF}{\mbox{$\mbox{\rm FF}$}}
\newcommand{\ALC}{$\cal ALC$}
\newcommand{\UF}{\mbox{$\mbox{\rm UF}_1$}}
\newcommand{\FAUF}{\mbox{$\forall\mbox{\rm -UF}$}}
\newcommand{\MK}{\mbox{$\mbox{\rm K}$}} \newcommand{\DMK}{\mbox{$\overline{\mbox{\rm K}}$}} \newcommand{\MDK}{\mbox{$\mbox{\rm DK}$}} \newcommand{\DMDK}{\mbox{$\overline{\mbox{\rm DK}}$}} \newcommand{\UNFO}{\mbox{$\mbox{\rm UNFO}$}}
\newcommand{\GNFO}{\mbox{$\mbox{\rm GNFO}$}}
\newcommand{\TGF}{\mbox{$\mbox{\rm TGF}$}}
\newcommand{\ZOIQ}{\mbox{$\mbox{$\mathcal{ZOIQ}$}$}}
\newcommand{\ZOQ}{\mbox{$\mbox{$\mathcal{ZOQ}$}$}}
\newcommand{\ZIQ}{\mbox{$\mbox{$\mathcal{ZIQ}$}$}}
\newcommand{\ZOI}{\mbox{$\mbox{$\mathcal{ZOI}$}$}}

\newcommand{\DMKPLUS}{\mbox{$\overline{\mbox{\rm K}}+$}} \newcommand{\Skolem}{\DMK\mbox{$\mbox{\rm-Skolem}$}} \newcommand{\ParamSkolem}[1]{\small \mbox{$\Skolem^{\forall=#1}$}} \newcommand{\ParamDMK}[1]{\small \mbox{$\DMK^{\forall=#1}$}} 

\newcommand{\GFcp}{\mbox{$\mbox{\rm GF}^{\times_2}$}}

\newcommand{\NLogSpace}{\textsc{NLogSpace}}
\newcommand{\NP}{\textsc{NPTime}}
\newcommand{\PTime}{\textsc{PTime}}
\newcommand{\PSpace}{\textsc{PSpace}}
\newcommand{\ExpTime}{\textsc{ExpTime}}
\newcommand{\ExpSpace}{\textsc{ExpSpace}}
\newcommand{\NExpTime}{\textsc{NExpTime}}
\newcommand{\TwoExpTime}{2\textsc{-ExpTime}}
\newcommand{\TwoNExpTime}{2\textsc{-NExpTime}}
\newcommand{\ThreeNExpTime}{3\textsc{-NExpTime}}
\newcommand{\APSpace}{\textsc{APSpace}}

\newcommand{\str}[1]{{\mathfrak{#1}}}
\newcommand{\restr}{\!\!\restriction\!\!}

\newcommand{\N}{{\mathbb N}}   \newcommand{\Q}{{\mathbb Q}}   \newcommand{\Z}{{\mathbb Z}}   

\newcommand{\cutout}[1]{}

\newcommand{\nb}[1]{}

\newcommand{\sss}{\scriptscriptstyle}

\newcommand{\Qfr}{\text{\rm Q}}

\newcommand{\Tree}{\mathbb{T}}
\newcommand{\treet}{\mathfrak{t}}

\newcommand{\tpstar}{{\rm tp}^*}
\newcommand{\tp}{{\rm tp}}
\newcommand{\type}[2]{{\rm tp}^{{#1}}({#2})}
\newcommand{\outertype}[2]{{\rm otp}^{{#1}}({#2})}
\newcommand{\hulltype}[2]{{\rm htp}^{{#1}}({#2})}

\newcommand{\Structflat}{Struct(\flat)}
\newcommand{\Assignflat}{Assign(\flat)}
\newcommand{\SpecPreStrflat}{SpecPreStr(\flat)}
\newcommand{\SpecAssignflat}{SpecAssign(\flat)}

\newcommand{\ax}{{\rm Ax}}
\newcommand{\tr}{{\rm Tr}}

\newcommand{\rs}{\bar{r}}
\newcommand{\vs}{\bar{v}}
\newcommand{\xs}{\bar{x}}
\newcommand{\ys}{\bar{y}}
\newcommand{\zs}{\bar{z}}
\newcommand{\cs}{\bar{c}}

\newcommand{\grade}{\textrm{grade}}
\newcommand{\atoms}{\textrm{atoms}}
\newcommand{\freevars}{\textrm{freevar}}
\newcommand{\Cons}{\textrm{Cons}}
\newcommand{\Rels}{\textrm{Rels}}
\newcommand{\Vars}{\rm{Vars}}
\newcommand{\arity}{\textrm{ar}}
\newcommand{\mex}{\textrm{mex}}
\newcommand{\proj}{\textrm{proj}}
\newcommand{\rrightarrow}{\mathrel{\mathrlap{\rightarrow}\mkern1mu\rightarrow}}
\newcommand{\SAT}{\rm{SAT}}
\newcommand{\VER}{\rm{VER}}
\newcommand{\GAME}{\textrm{G}}
\newcommand{\POS}{\textrm{POS}}
\newcommand{\chc}{\textrm{chc}}
\newcommand{\hull}{\textrm{hull}}

\newcommand{\bK}{K}
\newcommand{\bM}{M}

\newcommand{\eqdef}{:=}

\newcommand{\F}{\mathbb{F}}
\newcommand{\Field}{\textrm{GF}}

\newcommand{\poly}{\textrm{poly}}

\newcommand{\Perms}{\textrm{Perms}}
\newcommand{\tuple}[1]{\langle{#1}\rangle}

\newcommand{\arc}{\rightarrow}
\newcommand{\CColour}{\textrm{Control}_{\mu}}

\newcommand{\CRole}{\textrm{Control}_{\mu}}
\newcommand{\CBit}{\textrm{Control}_{\textrm{bit}}}
\newcommand{\Class}{\textrm{Class}}
\newcommand{\mask}{\textrm{repr}}
\newcommand{\select}{\textrm{select}}

\newcommand{\Cell}{\textrm{Cell}}

\newcommand{\bit}{\textrm{bit}}

\newcommand{\omegav}{\omega_{\textrm{V}}}
\newcommand{\omegas}{\omega_{\textrm{S}}}

\let\oldemptyset\emptyset
\let\emptyset\varnothing

\newcommand{\eqsiema}{\stackrel{\text{def}}{=\joinrel=}}

\title{On the complexity of Maslov's class $\overline{\text{K}}$}
\titlenote{Extended version of a LICS'24 paper.}

\author{Oskar Fiuk}
\orcid{0009-0006-1312-4899}
\affiliation{\institution{Institute of Computer Science, University of Wrocław}
\city{Wrocław} 
\country{Poland}
}
\email{307023@uwr.edu.pl}

\author{Emanuel Kiero\'nski}
\orcid{0000-0002-8538-8221}
\affiliation{\institution{Institute of Computer Science, University of Wrocław}
\city{Wrocław} 
\country{Poland}
}
\email{emanuel.kieronski@cs.uni.wroc.pl}

\author{Vincent Michielini}
\orcid{0000-0002-1413-9316}
\affiliation{\institution{Faculty of Mathematics, Informatics, and Mechanics, Warsaw University}
	\city{Warsaw}
\country{Poland}
}
\email{michielini@mimuw.edu.pl}

\begin{abstract}
Maslov's class \DMK{} is an expressive fragment of First-Order Logic known to have decidable
satisfiability problem, whose exact complexity, however, has not been established so far. 
We show that $\DMK$ has the exponential-sized model property, and hence its satisfiability problem is \NExpTime-complete.
Additionally, we get new complexity results on related fragments studied in the literature, and
propose a new decidable extension of the uniform one-dimensional fragment (without equality).
Our approach involves a use of satisfiability games tailored to \DMK{} and a novel application of paradoxical tournament graphs.
\end{abstract}


\maketitle

\newtheorem{claim}[theorem]{Claim}
\newtheorem{fact}[theorem]{Fact}

\section{Introduction}

Identifying elegant fragments of First-Order Logic with decidable satisfiability problem and good expressive power
has been an important theme in mathematics and computer science for decades.
Its motivations come from various areas, including
hardware and software verification, artificial intelligence, distributed computing, knowledge representation, databases and more.

In this line of research an interesting fragment was proposed by Maslov \cite{Mas71}. Originally, Maslov called his fragment  \MK{} and considered its 
validity problem (``Is a given sentence of the fragment true in all structures?''). Here, as in later works, we consider the dual of \MK{}, denoted by \DMK{}, and its satisfiability problem (``Is a given sentence of the fragment true in some structure?'').
\footnote{The reader should not confuse the class \DMK{} with  another class, called just the \emph{Maslov class}. The latter  consists of sentences of the shape $\exists^* \forall^* \exists^*.~\psi$, with $\psi$ 
being a quantifier-free Krom formula without equality, and is quite well understood (see, e.g.~\cite{BGG}).}

Converted to prenex form, \DMK-formulas  are as follows:
\[\forall x_1 \ldots \forall x_\bK.~\Qfr_{1}y_{1} \ldots \Qfr_\bM y_\bM.~\psi\]
where $\Qfr_i$'s are quantifiers, $\psi$ is
a quantifier-free formula without the equality symbol nor function symbols of arity greater than $0$ (constants are allowed), 
and every atom of $\psi$ satisfies one of the following conditions: (i) it contains at most one variable, (ii) its 
variables are precisely $x_1, \ldots, x_\bK$, or (iii) it contains an existentially quantified variable $y_i$, and no
$y_j$ with $j>i$.

The class \DMK{}   embeds, either syntactically or via standard reductions preserving satisfiability,  many known decidable fragments of First-Order Logic, including the monadic class \cite{Low15}, the Ackermann fragment \cite{Ack28} and its generalised version \cite{Voi19},   the G\"odel class \cite{God33}, the two-variable fragment \cite{Sco62}, Class~2.4 from \cite{DG79} (a solvable Skolem class, which in this paper we will call \Skolem).
It also captures basic modal logic and many standard description logics, e.g., the description logic $\mathcal{ALC}$ with
Boolean combinations of roles, inversions, role restrictions and positive occurrences of role compositions.
In this context it is considered (together with  the two-variable fragment, the guarded fragment and the fluted fragment) in  the survey \cite{HSG04}.
Even more formalisms, for example the uniform one-dimensional fragment \cite{HK14} or its variation with alternation of quantifiers in blocks \cite{Kie23}, are captured by the class  \DMDK{} of conjunctions 
of \DMK-sentences, also known to have decidable satisfiability \cite{HuS99}.

Maslov proved the decidability of the validity problem for \MK{}, which is equivalent to the satisfiability problem for \DMK{},
using his own approach, which he called the \emph{inverse method} (see~\cite{Lifschitz89}). There were a few subsequent works \cite{FLTZ93,HuS99,FLHT01}, whose authors reproved this result by means of the resolution method; all of them
work directly with \DMK{}. 
None of those works, however, studied the complexity of its satisfiability.
Even though it was hypothesised that the class \DMK{} may have non-elementary complexity {\cite{Voi19}, it seems that some elementary 
upper bound could be extracted from the resolution-based procedure in \cite{HuS99} (see Section \ref{s:discussion}).
This bound, however, could not be better than doubly exponential, which would still leave a gap, as the best lower bound inherited from
the known fragments embeddable in \DMK{} is \NExpTime-hardness, e.g., this lower bound holds already for the prefix-class $\forall\forall\exists$ \cite{Fur83}.
Also, we remark that the decidability of \DMK{} does not survive when equality is allowed, as this prefix-class $\forall\forall\exists$ becomes undecidable with equality \cite{Gol84}.

\medskip
\noindent
{\bf Our contribution.} In this paper we add the main missing brick to the understanding of \DMK{} by showing that its satisfiability problem is \NExpTime-complete.
We will do it by demonstrating that \DMK{} has the exponential-sized model property:

\begin{restatable}{theorem}{dmkexposize}
	\label{theorem:dmk-expo-size}
	Every satisfiable formula $\phi$ in $\DMK$ admits a finite model of size $2^{\mathcal{O}(|\phi|\cdot \log|\phi|)}$.
	Hence, the satisfiability problem for \DMK{} is \NExpTime-complete.
\end{restatable}

In contrast to the previous works on \DMK{}, which approached the problem \emph{syntactically}, we do it
\emph{semantically}. In particular, we use a game-theoretic view on the problem, and, more importantly, 
employ an adaptation of the results on the existence of small paradoxical tournament graphs. Up to our knowledge, a use of paradoxical tournaments
is novel in this area.

Our results transfer to \DMDK{} and entail the \NExpTime-complexity of two subfragments of \DMK{} studied in the literature, whose complexity has not been known so far: the Generalised Ackermann class (without equality) and the \Skolem{} class, the latter being the intersection of \DMK{} and the prefix-class $\forall^*\exists^*$ (the Skolem class).
In addition, we propose a new decidable fragment, being a generalisation of the uniform one-dimensional fragment from~\cite{HK14}.

The organisation of the paper is as follows:

In Section~\ref{section:background}, we give a formal definition of Maslov's class \DMK{} and introduce the various notions which we will be using in the paper.

In Section~\ref{section:graphs}, inspired by classical results on paradoxical tournament graphs, we introduce a variant of tournament graphs with colours of vertices and of arcs, and define for them an appropriate notion of paradoxicality. We show the existence of paradoxical such tournaments of small size.
We propose two constructions: a randomised one and a deterministic one. Both are rather routine adaptations of the classical constructions for the variant without colours. These tournaments will be the core of our exponential-sized models.  

In Section~\ref{section:satgames}, we introduce our \emph{satisfiability games} for \DMK{}.
We link the satisfiability of $\DMK$ to strategies in these games and obtain initial technical results concerning the number of $1$-types.

In Section~\ref{section:models}, we use colourful paradoxical tournaments from Section~\ref{section:graphs} and strategies in games from Section~\ref{section:satgames} to establish that satisfiable sentences in \DMK{} admit models of size $2^{\cO(|\phi|\cdot \log|\phi|)}$. Thus, we prove Theorem~\ref{theorem:dmk-expo-size}, the main result of our paper.

In Section~\ref{section:conjunctions}, we show that our upper bound on the size of minimal models (and hence also on the complexity) for \DMK{} extends to \DMDK{}.

In Section~\ref{section:lower-bounds}, we show that the obtained upper bound is essentially the best possible. We do this by supplying a corresponding family of tight examples:
for every $n \ge 3$, we construct a sentence $\phi_n$ such that it has size linear in $n$ and is satisfiable only in models of size at least $2^{\Omega(n\cdot \log n)}$.
Moreover, our formulas can be even assumed to be in the fragment \Skolem{}, without constant symbols and with just a single existential quantifier.  

In Section~\ref{section:skolem}, we show that satisfiable sentences in \DMK{} even have models of size $2^{O(|\phi|)}$, under the assumption that the number of universal quantifiers is
bounded. This in particular applies to the G\"odel class, admitting two universal quantifiers (see Section~\ref{s:discussion}).

In Section~\ref{section:application}, we propose a novel generalisation of the uniform one-dimensional fragment of First-Order Logic (without equality) \cite{HK14}: the $\forall$-uniform fragment.
We obtain the exponential-sized model property and \NExpTime-completeness of this new fragment by reducing its satisfiability to satisfiability of \DMDK{}.

In Section~\ref{s:discussion}, we conclude our paper by providing comments on previous works regarding Maslov's class \DMK{} and related fragments.

We believe that our results are valuable not only because they establish the precise bounds for the complexity and the size of minimal models for \DMK{} and related logics, but also because they give us a deeper
 understanding of the reasons behind them. Especially interesting is the boundary between $2^{\cO(|\phi|)}$ and $2^{\cO(|\phi|\cdot\log|\phi|)}$, which we solve by showing that the gap is crossed by the unboundedness of the number of universal quantifiers in the fragment $\Skolem$.

\section{Technical background}
\label{section:background}

In this section, we formally introduce Maslov's class $\DMK$, as well as the different notions needed for the proofs of the article.

We assume that the reader is familiar with the syntax of First-Order Logic ($\FO$).
We work with signatures containing relation and constant symbols, but no function symbols of arity greater than zero. Relation symbols may have arbitrary arities, including zero.
When building formulas, we allow standard Boolean connectives: $\vee$, $\wedge$, $\neg$ and $\rightarrow$, but we do not allow equality (unless explicitly stated).

\medskip
\noindent
\textbf{Naming conventions.}
We use Fraktur letters for structures and  the corresponding Roman capitals for their domains.
We usually use letters $a,b$ to denote elements of structures, $x,y,z$ for variables, and $c$ for constants; all of these possibly with some decorations and with a bar to denote tuples.
For a tuple of variables $\xs$, we use $\phi(\xs)$ to denote that all the free variables of $\phi$ are in $\xs$.

Sometimes it will be convenient to identify a tuple $\xs=\langle x_1,\ldots, x_k\rangle$ of variables with the corresponding set $\{x_i\mid 1 \leq i \leq k\}$; and therefore we will allow ourselves notations such as $y\notin\xs$, $\xs\cup\{y\}$, etc. We will keep the same convention for any kind of tuples (tuples of natural numbers, tuples of elements of a structure, etc.).

We write $\N$ for the set of natural numbers $\{0,1,\dots\}$, and, if $k \in \N$, $[k]$ denotes the set~$\{1, \ldots, k\}$ (in particular if $k=0$ then it is the empty set $\emptyset$).

\medskip
\noindent
\textbf{Measuring size.}
By the \emph{size of a structure}, we simply mean the cardinality of its domain.
By $|\phi|$ we denote the \emph{size of a formula} $\phi$ measured in the uniform way: write $\phi$ as a word over the alphabet consisting of quantifiers, Boolean connectives, comma, parentheses, variables, relation and constant symbols; then each occurrence of a symbol contributes as $1$ to the size.

We point out that other authors might measure the size of formulas in bits.
Hence, one should be careful when comparing results from different works.

\subsection{The fragments $\overline{\text{K}}$, $\overline{\text{DK}}$ and $\overline{\text{K}}$-Skolem}

\medskip
\noindent
\textbf{Maslov's class \DMK.} 
Let $\phi$ be a sentence in negation normal form, and let $\gamma$ be one of its atoms.
The \emph{$\phi$-prefix} of $\gamma$ is the sequence of quantifiers in $\phi$ binding the variables of $\gamma$. For instance, if $\varphi$ is the sentence $\exists x.~\forall y.~\exists z.~R(x,y)\wedge T(c, y, z, y)$, $c$ being a constant symbol, then the $\phi$-prefix of the atom $R(x,y)$ is the sequence ``$\exists x.~\forall y$'', while that of the atom $T(c, y, z, y)$ is  ``$\forall y.~\exists z$''. An atom without variables (e.g. talking only about constants) has an empty $\varphi$-prefix. 

The class  $\DMK$ consists of the sentences $\varphi$ which are in negation normal form and in which there exist
universally quantified variables $x_1, \ldots, x_K$, none of which lies within the scope of any existential quantifier,
such that each atom of $\phi$ has a $\phi$-prefix of one of the following shapes: 
\begin{itemize}[nosep]
	\item a $\phi$-prefix of length at most $1$,
	\item a $\phi$-prefix ending with an existential quantifier,
	\item or exactly the sequence ``$\forall x_1\ldots\forall x_\bK$''.
\end{itemize}

The variables $x_1, \ldots, x_\bK$ are called the \emph{special variables} of the formula $\varphi$, and their number $\bK$ is called the \emph{grade} of $\varphi$.

The reason we need the formula $\varphi$ to be in negation normal form is due to the antisymmetry of the definition with respect to the universal and existential quantifiers. In a formula such as $\neg\exists x.~P(x)$, the variable $x$ is quantified existentially, but its semantic role would actually be universal. However, for convenience, we can allow a relaxed definition where the formula $\varphi$ is not asked to be in negation normal form, but only to not have negations binding quantifiers. We can then allow the use of the implication symbol $\rightarrow$, as long as its left-hand side does not contain any quantifier. With this convention, the following formula $\varphi_{\textrm{co\_authors}}$ is in $\DMK$:
\begin{flalign*}
	\forall s_1,&s_2, s_3.
	\big[\textrm{scientist}(s_1)\wedge\textrm{scientist}(s_2)\wedge\textrm{scientist}(s_3)\\
	&\wedge\textrm{co\_authors}(s_1, s_2, s_3)\big]\\
	&\rightarrow \exists a.~\textrm{article}(a)\wedge \textrm{written\_by}(a, s_1,s_2,s_3).
\end{flalign*}

Indeed, the $\varphi_{\textrm{co\_authors}}$-prefixes of the different atoms are: the singleton sequences ``$\forall s_1$'', ``$\forall s_2$'', ``$\forall s_3$'', and ``$\exists a$''; the sequence ``$\forall s_1.~\forall s_2.$ $\forall s_3.~\exists a$'', which ends with an existential quantifier; and the universal sequence ``$\forall s_1.~\forall s_2.~\forall s_3$''. Since there is no existential quantifier binding the quantifiers $\forall s_i$, all the conditions are met: $\varphi_{\textrm{co\_authors}}$ is a formula in $\DMK$ of grade~$3$, with its special variables being $s_1$, $s_2$ and $s_3$.

As a second example, we consider the formula~$\varphi_{\textrm{marriage}}$:
\begin{flalign*}
	\forall  h,&w.~\textrm{husband\_and\_wife}(h,w)\rightarrow\\ 
	\exists & p.~\textrm{problem}(p)\wedge \forall d.~\textrm{date}(d)\rightarrow\\ 
	& \exists d'.~\textrm{date}(d') \wedge \textrm{later\_than}(d',d) \wedge \textrm{occurs\_to\_at}(p,h,w,d').
\end{flalign*}

This formula is of grade $2$, with $h$ and $w$ being its special variables. Although the variable $d$ is quantified universally, it is not special.

On the contrary, an example of a first-order sentence \emph{not} belonging to $\DMK$ is the axiom of transitivity:\[\varphi_{\textrm{trans}}\eqdef\forall x, y, z.~\big[T(x,y)\wedge T(y, z)\big]\rightarrow T(x,z).\]

Indeed, the reader can see that no subset of $\{x,y,z\}$ is a legitimate candidate for being the special variables of $\varphi_{\textrm{trans}}$.

Using standard procedures, we can convert any sentence in~$\DMK$ into its prenex form and move the quantifiers $\forall x_1, \ldots, \forall x_\bK$ to the front.
This way we obtain sentences as follows:
\begin{eqnarray}
	\label{eq:prenex}
	\forall x_1 \ldots \forall x_\bK.~\Qfr_{1}y_{1} \ldots \Qfr_\bM y_\bM.~\psi,
\end{eqnarray}
where $\psi$ is quantifier-free. For the rest of this paper, we will work with formulas of this shape, assuming without loss of generality that $\bK>0$, i.e. that the first quantifier is universal.

In the literature, one can find definitions of \DMK{} allowing an extra initial prefix of existential quantifiers. In our version,
we can simulate them via the use of constant symbols.

\medskip
\noindent
\textbf{The class \DMDK{}.} In our work, we also consider the class \DMDK{} consisting of all (finite) conjunctions of sentences from \DMK{}.
Notice that a formula $\phi$ in \DMDK{} might be a conjunction of sentences with different grades. In such a case, a priori $\phi$ is not equivalent to any formula in~$\DMK$.

\medskip
\noindent
\textbf{The class \DMK-Skolem.} The class \Skolem{} is the intersection of \DMK{} and the \emph{Skolem class}, the latter being the set of prenex formulas with quantifier prefixes
of the form $\forall \bar{x}~\exists \bar{y}$. In effect we can assume that formulas in \Skolem{} have the shape $\forall \bar{x}~\forall \bar{z}~\exists\bar{y}.~\psi$, with $\bar{x}$ being the tuple of special variables. 
Originally \Skolem{} was introduced by Dreben and Goldfarb in the book \cite{DG79} under the name Class 2.4, without any connection to \DMK.

If we convert our example sentence $\varphi_{\textrm{co\_authors}}$ to prenex form, we indeed get a sentence in \Skolem{}.
Our second example $\varphi_{\textrm{marriage}}$ goes beyond it,  as it contains an alternation of quantifiers.

It is worth to mention two important fragments of \Skolem{}: the Ackermann class and the G\"odel class.
The former consists of prenex sentences with quantifier prefixes of the form 
$\forall x~\exists \bar{y}$ (one universal quantifier), the latter---$\forall x_1, x_2~\exists \bar{y}$ (two universal quantifiers). 
Both are often presented in ``initially extended'' versions in which an additional prefix of existential quantifiers of arbitrary length
is admitted. Again, in our setting we do not need to consider such prefixes, as they can be naturally simulated by constants.

\medskip
One more class, the uniform one-dimensional fragment, \UF{}, will be relevant for us. As it does not play a central role in this paper, we will
define it later, in Section \ref{section:application}, dedicated to a generalisation of it.

\subsection{Semantics: formal definitions}

A \emph{signature} is a tuple $\sigma = (\Cons,\Rels,\arity\colon \Rels \rightarrow \N)$ where $\Cons$ and $\Rels$ are sets of constant and relation symbols respectively. The number $\arity(R)$ is called the \emph{arity} of the symbol $R$.
By $\sigma(\phi)$ we denote the signature consisting of relation and constant symbols mentioned in $\phi$.

We call \emph{structure} over the signature $\sigma$ any tuple $\str{A}=(A^\Cons, R^\str{A}\colon (A^{\Cons})^{\arity(R)} \rightarrow \{0,1\})_{R \in \Rels}$,
where $A^\Cons=A \cup \Cons$ and $A \cap \Cons = \emptyset$.
We say that: $A^\Cons$ is the \emph{domain}, $A$ is the \emph{unnamed domain}, its elements being the \emph{unnamed elements}.

We do not include a function interpreting the constant symbols inside the domain, but rather assume that constant symbols of $\sigma$ are interpreted by themselves. In particular, this means that different constant symbols are interpreted distinctly.

However, in the context of satisfiability, this does not affect the generality of our results: 
we can non-deterministically guess a partition of $\Cons$, corresponding to the equalities among the interpretations of constants, and substitute the occurrences of the constant symbols from each group by its fixed representative, hence reducing the problem to our scenario (see Appendix \ref{appendix-constants}).

Any function $f\colon X \rightarrow A^\Cons$ such that $A \subseteq f(X)$ is denoted by $f\colon X \rrightarrow A^\Cons$.
If $B \subseteq A$, we denote by $\str{A} \restr B$ the restriction of $\str{A}$ to its \emph{subdomain} $B \cup \Cons$.

We also use \emph{partial structures}, in which some relations may not be defined on some tuples. This is 
captured by extending the range of every function $R^\str{A}$ to $\{0,1,\bot\}$, the symbol $\bot$ standing for ``undefined''. 
If $\str{A}$ is a partial structure, whenever we write $\str{A}\models\phi$, we ensure that all the information necessary to determinate the truth value of the sentence $\phi$ is indeed defined.

In our proofs, we will make a big use of different versions of types: \emph{$1$-types}, \emph{$k$-outer-types} and $k$-\emph{hull-types}.

An (atomic) $1$-\emph{type} over a signature $\sigma$ is any $\sigma$-structure with the domain $[1]^\Cons=\{1\} \cup \Cons$. Notice that, in general, the number of $1$-types over $\sigma$ is doubly exponential in $|\sigma|$, as, if $\sigma$ admits a constant symbol $c$ and a relational symbol $R$ of arity $n$, then there are $2^n$ possible tuples consisting of $1$'s and $c$'s, and therefore at least $2^{2^n}$ possible functions from $([1]^\Cons)^{n}$ to $\{0,1\}$.
A \emph{$0$-type} over $\sigma$ is any $\sigma$-structure with the domain $\emptyset^\Cons$ consisting of only the constants.

In this paper, $k$-types, for $k\geq 2$, will not be needed, due to the syntax limitations of $\DMK$. We will yet make use of a relaxed version of $k$-types, namely $k$-\emph{outer-types}, where the relations are defined only for certain tuples.

For $k \ge 0$, a $k$-\emph{outer-type} over $\sigma$ is a partial structure $\str{B}$ of domain~$[k]^\Cons$,
in which, for every $R$ and every $\bar{a} \in (A^\Cons)^{\arity(R)}$,
$R^\str{B}(\bar{a})$ is defined (i.e. its  value is $0$ or $1$) if and only if the intersection $\bar{a}\cap[k]$ is the full set $[k]$ or has at most one element.

Finally, a \emph{$k$-hull-type} over $\sigma$ is a partial structure $\str{B}$ of domain $[k]^\Cons$ in which, for every $R$ and every tuple $\bar{a} \in (A^\Cons)^{\arity(R)}$, $R^\str{B}(\bar{a})$ is defined if and only if $\bar{a}\cap[k]=[k]$.

The reader will notice that the notions of $1$-types, $1$-outer-types and $1$-hull-types coincide. An \emph{outer-type} (resp. a \emph{hull-type}) is a $k$-outer-type (resp. a $k$-hull-type) for some $k \ge 0$. We call $k$ the \emph{grade} of this type.
In the paper we will use $\alpha$ and $\beta$ to denote $1$-types and outer-types respectively; possibly with decorations.

We say that a set $\BBB$ of outer-types is \emph{consistent} if it induces a unique $0$-type, i.e. if for all $\beta_1,\beta_2 \in \BBB$
we have $\beta_1 \restr \emptyset = \beta_2 \restr \emptyset$. By $\AAA$ we denote the subset of $\BBB$ consisting of all the $1$-types it contains.

Let $\str{A}$ be a partial structure and let $a\in A$ be an unnamed element. We denote by $\tp^\str{A}(a)$ the partial $\sigma$-structure of domain $[1]^\Cons$ which is isomorphic to $\str{A}\restr \{a\}$ via the mapping $a\mapsto 1$ and $\Cons \ni c\mapsto c$. It is a $1$-type when, for every relational symbol $R$ and every tuple $\bar{b}\in(\{a\}\cup\Cons)^{\arity(R)}$, $R^\str{A}(\bar{b})$ is indeed defined. In this case we call it the $1$-type \emph{realised} by $a$. Every time we refer to some $\tp^\str{A}(a)$ in the paper, it will indeed be a $1$-type. In particular, if $\beta$ is a $k$-outer-type, for $k\geq 1$, then $\tp^\beta(i)$ is a $1$-type.

Similarly, if $\bar{a} \in A^k$ is a tuple of $k$ pairwise distinct unnamed elements, we can define $\outertype{\str{A}}{\bar{a}}$ and $\hulltype{\str{A}}{\bar{a}}$ in the analogous way, as the partial $\sigma$-structure of domain $[k]^\Cons$ which is isomorphic to $\str{A}\restr \bar{a}$ via the mapping $a_i \mapsto i$ and $\Cons \ni c \mapsto c$. Again these are respectively $k$-outer-types and $k$-hull-types when the relations are defined for the according tuples. It could even be that the relations are defined for ``too many tuples'' (for instance, if $R(a_1, a_3)$ is defined, with $k=3$). In this case, we assume that the tuples not needed for the definitions are set to $\bot$, in order for us to get $k$-outer-types or $k$-hull-types. Again, in the whole paper, when we call for some $\outertype{\str{A}}{\bar{a}}$ or some $\hulltype{\str{A}}{\bar{a}}$, they will always indeed be $k$-outer-types or $k$-hull-types respectively.

\section{Paradoxical tournament graphs with colours}
\label{section:graphs}

In this section, we introduce a new combinatorial notion, \emph{paradoxical colourful tournaments}: their structure will serve as the foundation of exponential-sized models for satisfiable formulas in $\DMK$.

This notion generalises an already known notion of \emph{paradoxical tournaments}, of which we recall the definition here.

Let $\cG=(V,E)$ be a directed graph. We write $a \arc b$ for $(a, b)\in E$ (and say that there is an \emph{arc} from $a$ to $b$). 
We say that $\cG$ is a \emph{tournament} if it does \emph{not} admit \emph{self-loops} (i.e. arcs $a\arc a$), and if there is exactly one directed arc between each pair of distinct vertices.

Let $\cT$ be a tournament. Let $A$ be a subset of vertices, and let $b$ be a vertex not in $A$. Then $b$ \emph{dominates} $A$ if for each $a \in A$ we have $b \arc a$.
Let $k\in\N$, a tournament is called \emph{$k$-paradoxical} if for each subset $A$ of at most $k$ vertices there exists a vertex dominating $A$.

It is a classical result by Erd\H{o}s that such tournaments exist~\cite{Erds1963OnAP}. By applying a probabilistic method, he obtained that there are $k$-paradoxical tournaments of size $\cO(k^2 2^k)$; also, he provided a lower bound $\Omega(2^k)$.

\medskip
\noindent
{\bf Paradoxical colourful tournaments.}
Let $\cR$ be a set of vertex colours, and let $\cQ$ be a set of arc colours.
We always assume that both $\cR$ and $\cQ$ are non-empty and finite.
Let $\cT = (V,E)$ be a tournament with a pair of labellings $\mu\colon V \rightarrow \cR$ and $\lambda\colon E \rightarrow \cQ$.
The triple $(\cT,\mu,\lambda)$ is called an $(\cR,\cQ)$-\emph{colourful tournament} (or more simply a \emph{colourful tournament}).

We define now a paradoxical notion for colourful tournaments. The triple $(\cT,\mu,\lambda)$ is said to be \emph{$(\cR,\cQ)$-paradoxical} if it admits the following property, where $\ell=|\cQ|$ is the number of arc colours:
for any vertex colour $r \in \cR$, any tuple $\bar{a}=\tuple{a_1,\dots,a_\ell}$ of pairwise distinct vertices,
and any tuple $\bar{q}=\tuple{q_1,\dots,q_\ell}$ of (non-necessarily distinct) arc colours,
there exists a vertex $b$ such that:
\begin{itemize}[nosep]
	\item $b$ dominates $\{a_1,\dots,a_\ell\}$ (in particular, $b\not=a_i$, for all $i$);
	\item $\mu(b)=r$;
	\item $\lambda(b \arc a_i)=q_i$, for every $i\in[\ell]$.
\end{itemize}
We say that such a vertex $b$ \emph{colourfully dominates} $\bar{a}$ via $r$ and $\bar{q}$.
In our use of the definition above, we will consider only non-trivial cases, i.e. the size of $\cT$ is at least $\ell$.

We now prove the existence of paradoxical colourful tournaments, and argue that the obtained bound is essentially optimal.
The proof is a direct extension of the original probabilistic proof by Erd\"os.

\begin{lemma}\label{lemma:colourful-tournaments}
	Let $\cR$ be a set of vertex colours and $\cQ$ be a set of arc colours. Then there exists an $(\cR, \cQ)$-paradoxical colourful tournament of size $2^{\cO(|\cQ|\cdot\log |\cQ|)}\times |\cR|\cdot\log|\cR|$.
	
	Moreover, the size of any $(\cR, \cQ)$-paradoxical colourful tournament is at least $2^{\Omega(|\cQ|\cdot\log |\cQ|)}\times|\cR|$.
\end{lemma}

\begin{proof}
	Let $n \in \N$ be a free parameter, supposed to be at least $|\cQ|$, which we denote here by $\ell$ for convenience.
    Let $V = \cR \times [n]$. Define the labelling $\mu\colon V\ni(r,i)\mapsto r\in \cR$. Let $\mathbb{T}$ be the set of all possible tournaments having $V$ as their set of vertices.
	Consider now a tournament $\cT = (V,E)$ from $\mathbb{T}$, and a labelling $\lambda\colon E \rightarrow \cQ$.

	Let us fix a vertex colour $r \in \cR$, a tuple $\bar{a}=\tuple{a_1,\dots,a_\ell}$ of pairwise distinct vertices, and a tuple $\bar{q}=\tuple{q_1,\dots,q_\ell}$ of arc colours.
	For a vertex $b \in (\{r\} \times [n]) \setminus \bar{a}$, we denote by $\cE_{r,\bar{a},\bar{q}}(b)$ the event, in a probabilistic sense, that $b$ colourfully dominates $\bar{a}$ via $r$ and $\bar{q}$.
The probability that $b$ dominates $\bar{a}$ (without considering colours) is exactly $\delta_1=2^{-\ell}$, and, if we assume it does, then 
	the probability that $\lambda(b\to a_i)=q_i$ for every $i\in[\ell]$ is exactly $\delta_2=\ell^{-\ell}$.
	Hence, the probability that the event $\cE_{r,\bar{a},\bar{q}}(b)$ holds does not depend on $b$ and is $\delta = (2\ell)^{-\ell}$.
	Moreover, for distinct vertices~$b$ and $b'$, the events $\cE_{r,\bar{a},\bar{q}}(b)$ and $\cE_{r,\bar{a},\bar{q}}(b')$ are independent, probabilistically speaking.
	
	Let $\cE_{r,\bar{a},\bar{q}}$ be the event $\bigcup_{b \in (\{r\} \times [n]) \setminus \bar{a}}{\cE_{r,\bar{a},\bar{q}}(b)}$, stating
	that there exists some vertex colourfully dominating $\bar{a}$ via $r$ and $\bar{q}$.
	We can bound its probability of not happening as follows:
	$$\Pr[\neg\cE_{r,\bar{a},\bar{q}}] \le (1{-}\delta)^{n-\ell} \le \exp(-(n{-}\ell)\cdot\delta).$$
	
	The second inequality coming from the fact that $n{-}\ell\geq 0$.
	Now, the number of triples $(r,\bar{a},\bar{q})$ can be bounded in the following way: $$ |\cR| \cdot {|V| \choose \ell} \cdot \ell^\ell \le |\cR|\cdot (|\cR|\cdot n)^\ell \cdot \ell^\ell \le (|\cR|\cdot n\cdot \ell)^{\ell+1}.$$
	
	Hence, if we denote by $p$ the value $\exp(-(n{-}\ell)\cdot\delta)$, and by $m$ the value $(|\cR|\cdot n\cdot\ell)^{\ell+1}$, the probability that $\cT$ is not a paradoxical colourful tournament, i.e. that some tuple $\bar{a}$ cannot be colourfully dominated via some role $r$ and some tuple $\bar{q}$, by any element of $V\setminus\bar{a}$, is at most $m\cdot p$.
	
	To show the existence of an $(\cR, \cQ)$-paradoxical colourful tournament, it is enough to find  $n$ for which $m\cdot p < 1$.
	Applying the logarithmic function to this inequality, we get the following:
	$$(\ell{+}1)\times \ln\big(|\cR|\cdot n\cdot \ell\big) < (n{-}\ell)\times \delta.$$
    It can be verified that for all $n \ge 10\cdot(2\ell)^{\ell+1}\cdot(\ln|\cR|+\ell{\cdot}\ln\ell)$ the inequality above holds. We can therefore conclude the exponential upper bound (remember that the number of vertices is $|\cR|\times n$).
	
	Now for the lower bound, let us consider an arbitrary $(\cR,\cQ)$-paradoxical colourful tournament $(\cT,\mu,\lambda)$.
	Choose a tuple $\bar{a}=\tuple{a_1,\dots,a_{\ell}}$ of pairwise different vertices of $\cT$.
	For every vertex colour $r \in \cR$ and every tuple $\bar{q}=\tuple{q_1,\dots,q_{\ell}}$ of arc colours, there must be a vertex $b_{r,\bar{q}}$ colourfully dominating $\bar{a}$ via $r$ and $\bar{q}$.
	Clearly, all the $b_{r,\bar{q}}$'s must be distinct, and we can deduce that the size of $\cT$ is at least $\ell!\times|\cR| = 2^{\Omega(\ell\log \ell)}\times|\cR|$.
\end{proof}

The proof of Lemma~\ref{lemma:colourful-tournaments} is non-constructive.
Yet, an explicit construction of paradoxical colourful tournaments is actually possible, we give it in Appendix~\ref{appendix:graphs}.
This explicit construction is based on Paley graphs and the result of Graham and Spencer \cite{graham_spencer_1971}, which states that Paley graphs of size $\Omega(k^2\cdot2^{2k})$ are $k$-paradoxical.

\section{Satisfiability games}
\label{section:satgames}

In this section, we show that the satisfiability of sentences in $\DMK$ can be studied via certain games.
First, we adapt a standard verification game for First-Order Logic, and then we introduce our satisfiability game tailored to the fragment $\DMK$.

In this whole section, we fix a sentence $\varphi$ in $\DMK$ of the shape as in (\ref{eq:prenex}), i.e. $\forall x_1 \ldots \forall x_\bK.~\Qfr_{1}y_{1} \ldots \Qfr_\bM y_\bM.~\psi$, with grade $K$ and special variables $x_1,\ldots, x_K$.
We denote by $\Cons$ and $\Rels$ the sets of respectively constants and relational symbols from $\sigma\eqdef\sigma(\phi)$,
by $\Vars=\xs\cup\ys$ the set of variables of $\phi$ and,
for each $0\leq i\leq \bM$, by ${\Vars}_i$ the set $\bar{x} \cup \{y_1, \ldots, y_i\}\subseteq\Vars$.

\subsection{Verification game}

Satisfaction of  $\phi$ in a given  $\sigma$-structure $\str{A}$  is naturally connected to the game $\VER(\phi, \str{A})$ between the \emph{existential player}, Eloisa,
trying to show that $\str{A} \models \phi$ and the \emph{universal player}, Abelard, trying to show the opposite.

A \emph{position} in the game is any assignment $f_t\colon {\Vars}_{t}\to A^\Cons$. The number $t$ is called the \emph{order} of $f_t$. The game, which has $M{+}1$ rounds, goes as follows. Abelard first chooses an assignment $f_0\colon\xs\rightarrow A^\Cons$ of order $0$
as he wishes to. In Round $t{+}1$, for $0 \le t \le M{-}1$, after a position $f_{t}\colon{\Vars}_{t}\to A^\Cons$ of order $t$ has been reached, the appropriate player (Abelard if $Q_{t+1}=\forall$, Eloisa if $Q_{t+1}=\exists$) extends it to $f_{t+1}$ by assigning $y_{t+1}$ to
an element of $A^\Cons$ of their choice.

At the end of the game, the players have constructed an assignment $f_M\colon \Vars \rightarrow A^\Cons$. Eloisa wins if eventually $\str{A}, f_M \models \psi$ (i.e. if the formula $\psi$ holds in $\str{A}$ when every variable in $\Vars$ is interpreted by its value via $f$). It is well known that $\str{A} \models \phi$ if and only if Eloisa has a winning strategy in the game $\VER(\phi, \str{A})$, in a sense which we will specify later.

\subsection{Satisfiability game}
Now, we introduce a more abstract  game $\SAT(\phi, \BBB)$ in which the structure $\str{A}$ is not given. Instead, it contains the parameter $\BBB$, which is a consistent set of
outer-types $\BBB$ meeting some closure conditions.  Eloisa tries to show that $\phi$ has a model in which all realised outer-types of grade at most $K{+}M$
are in $\BBB$. However, the entire model is not explicitly constructed. Rather, during the game, a partial structure $\str{L}$ and an assignment $f\colon \Vars \rrightarrow L^\Cons$ are constructed (remember that this notation means that the whole unnamed domain $L$ is included in $f[\Vars]$). Eloisa wins if finally $\str{L}, f \models \psi$.

The game is defined for $\BBB$ being  a consistent and \emph{closed} set of outer-types over $\sigma$, each of them having grade at most $K{+}M$.
We say that $\BBB$ is \emph{closed} if it satisfies the following conditions:
\begin{itemize}[nosep]
	\item for every outer-type $\beta \in \BBB$ and each $i\in[\grade(\beta)]$, the $1$-type  $\tp^{\beta}(i)$ is in $\BBB$;
    \item for every outer-type $\beta \in \BBB$ and every permutation $\pi$ of the set $[\grade(\beta)]$, the outer-type isomorphic to $\beta$ via $\pi$ is in $\BBB$; 
    \item for every $k\in[\bK{+}\bM]$ and every sequence $\alpha_1, \ldots, \alpha_k$ of $1$-types from~$\BBB$, there is a $k$-outer-type $\beta \in \BBB$ such that, for each $i\in[k]$, $\tp^\beta(i)=\alpha_i$.
\end{itemize}

The game consists of $M{+}1$ Rounds $0, 1, \ldots, \bM$. After Round $t$, we reach in the game a \emph{position} consisting of a pair $(\str{L}_t, f_t)$, where $\str{L}_t$ is a partial $\sigma$-structure,
with the unnamed domain $[k]$, for some $k \in \N$, and $f_t\colon{\Vars}_t \rrightarrow L_t^\Cons$ is an assignment.

\smallskip\noindent
{\bf Round 0.} Abelard chooses an outer-type $\str{L}_0 \in \BBB$, and an assignment $f_0\colon {\Vars}_0 \rrightarrow L_0^\Cons$.
Note that, as $\str{L}_0$ is an outer-type, its domain is indeed $[k]$, with $k=\grade(\str{L}_0)$. Moreover $k \le K$, since
the image $f_0[{\Vars}_0]$, which is of size at most $K$, must contain all the unnamed elements of $L_0$. It may happen that $k=0$, if $f_0$ maps all the variables to constants.

\smallskip\noindent
{\bf Round $\mathbf{t{+}1}$.}
We suppose that the position reached after the previous Round~$t$ is $(\str{L}_t,f_t)$, where $\str{L}_t$ is a partial $\sigma$-structure with
the domain $[k]^\Cons$, for some $k$, and $f_t\colon{\Vars}_t \rrightarrow L_t^\Cons$ is an assignment. Again, the number $t$ is called the \emph{order} of both the position $(\str{L}_t,f_t)$ and of the assignment $f_t$.

If $\Qfr_{t+1} = \forall$, then the move belongs to Abelard. His task is to assign an element to the variable $y_{t+1}$. He has two options.

\noindent
	(a) He may extend the domain by adding an unnamed element. In this case, $L_{t+1}$ is the set $[k{+}1]$, and he extends the assignment by defining $f_{t+1}$ as the function $f_t \cup \{y_{t+1}\mapsto k{+}1\}$. He then chooses a $1$-type $\alpha \in \BBB$ and defines the partial structure $\str{L}_{t+1}$ as follows: $\str{L}_{t+1}\restr L_t=\str{L}_t$, $\tp^{\str{L}_{t+1}}(k{+}1)=\alpha$, and the rest is undefined.
	
	\noindent
	(b) He may keep the domain $[k]^\Cons$, and choose an element $a_{t+1}$ in $L_t^\Cons$. In this case the structure $\str{L}_{t+1}$ is $\str{L}_t$, unchanged, and
	 $f_{t+1}$ is the function $f_t \cup \{y_{t+1}\mapsto a_{t+1}\}$.

Notice that, during his move, Abelard has no direct control over how its new chosen element $f_{t+1}(y_{t+1})$ interacts with the other elements of the domain $L_{t+1}^\Cons$. This is completely justified by the limited syntax of $\DMK$: if $y_{t+1}$ occurs in an atom $\gamma(\vs)$ of $\psi$, then either $\vs$ is actually the singleton $\{y_{t+1}\}$ (in which case the outer-type chosen by Abelard will set the truth value of the atom), or $\vs$ contains an existentially quantified variable $y_{t'}$, $t{+}1<t'$, and the assignment of this variable will be taken care of by Eloisa later in the game.

Hence, let us now describe Eloisa's turn, i.e. when $\Qfr_{t+1} = \exists$. On the contrary to Abelard, Eloisa has no control over the new unnamed domain $L_{t+1}$, which is the set $[k{+}1]$, nor on the assignment $f_{t+1}\colon{\Vars}_{t+1}\rrightarrow L_{t+1}^\Cons$, defined as the function $f_t \cup \{y_{t+1}\mapsto k{+}1\}$. Her role is to define a partial structure $\str{L}_{t+1}$ satisfying $\str{L}_{t+1}\restr L_t = \str{L}_t$, i.e. she really chooses only atoms containing the new element $k{+}1=f_{t+1}(y_{t+1})$. She does this in two steps.

First, she chooses a $1$-type $\alpha\in\BBB$ and sets $\tp^{\str{L}_{t+1}}(k{+}1)=\alpha$.
Then, for each subset $A$ of $L_{t+1}$ containing $k{+1}$, she selects an outer-type $\beta\in\BBB$ for it: $\outertype{\str{L}_{t+1}}{\bar{a}}=\beta$, with $\bar{a}$ being an enumeration of $A$ (note that $A$ has size at most $K{+}M$, the maximal grade of an outer-type of $\BBB$).
This outer-type $\beta$ shall be consistent with the already defined $1$-types of each individual $a$ in $A$ (the closure of $\BBB$ makes meeting this requirement always possible).
As the outer-types on distinct tuples share no atoms, except those corresponding to $0$- and $1$-types, no conflict are met when defining $\str{L}_{t+1}$.

\smallskip\noindent
{\bf Winning condition.} After Round $M$, Eloisa wins the game if $\str{L}_\bM, f_\bM \models \psi$. Note that in $\str{L}_\bM$, all the atoms, which are required to verify if $\psi$ is satisfied under $f_\bM$, are indeed defined.

\smallskip\noindent
{\bf Memoryless strategies.} 
Let $\GAME$ be either $\VER(\phi, \str{A})$ or $\SAT(\phi, \BBB)$. It can be easily checked that any position reached at any moment in $\GAME$ uniquely determines the positions reached before. For this reason, we can consider \emph{memoryless strategies}.
A \emph{(memoryless) strategy} for Eloisa in $\GAME$ is a function $\omega$ that, for every $0\leq t< M$ such that $\Qfr_{t+1}=\exists$, assigns to every position $\rho_t$ of order $t$ a next position $\rho_{t+1}$, in accordance with the rules of $\GAME$.
By $\cP_\omega^\exists(\GAME)$, we denote the set of positions of $\GAME$ which are obtainable after Eloisa's rounds when following the strategy~$\omega$.

If $\GAME$ is $\SAT(\varphi,\BBB)$ (resp. $\VER(\varphi,\str{A})$), then we say that the strategy $\omega$ is \emph{winning} if for every position $\rho_M=(\str{L}_M,f_M)$ (resp. $\rho_M = f_M\colon \Vars\rightarrow A^\Cons$) of order $M$ in $\cP_\omega^\exists(\GAME)$, we have $\str{L}_M, f_M\models\psi$ (resp. $\str{A},f_M \models \psi$).

Finally, if $\GAME$ is $\SAT(\varphi,\BBB)$, then $\cF_\omega^\exists(\GAME)$ denotes the set of assignments $f_t$ such that $(\str{L}_t,f_t)\in \cP_\omega^\exists(\GAME)$ for some structure~$\str{L}_t$.

As we already mentioned, it is well known that Eloisa has a winning strategy in  $\VER(\phi,\str{A})$ if and only if $\str{A}\models\phi$. Now, we want to connect   $\SAT(\phi,\BBB)$
to the (un)satisfiability of $\phi$. We first show that our notion of satisfaction games is \emph{complete}:

\begin{restatable}{lemma}{satandgames}
\label{l:satandgames}
Assume that $\phi$ is satisfiable. Then there is a consistent and closed set of outer-types $\BBB$ such that Eloisa has a winning strategy in $\SAT(\phi, \BBB)$.
\end{restatable}

\begin{proof}
Since the proof of this lemma is rather standard, we present only a sketch.

We begin by taking an arbitrary model $\str{A}_0 \models \varphi$ and augmenting it to $\str{A}$ by taking infinitely many copies of
every element, including unnamed copies of constants, and defining the structure of  $\str{A}$ in the symmetric way:
$\str{A} \models R(\bar{a})$ iff $\str{A}_0 \models R(\bar{a}_*)$, where $\bar{a}_*$ is obtained by replacing every element of $\bar{a}$ with its original copy from $\str{A}_0$.

As $\phi$ does not use equality it follows that $\str{A} \models \phi$.
Moreover, Eloisa can win $\VER(\phi, \str{A})$ using a \emph{proper} strategy,
that is, when extending the assignment, she always chooses a fresh unnamed element for the variable.

Let $\BBB$ be the set of outer-types of grade at most $\bK{+}\bM$ and realised in $\str{A}$. It is readily verified
that $\BBB$ is consistent and closed. To win $\SAT(\phi, \BBB)$, Eloisa fixes her winning proper strategy $\omegav$ for $\VER(\phi, \str{A})$, and, in parallel to $\SAT(\phi, \BBB)$, she simulates the game $\VER(\phi, \str{A})$ consistently with
  $\omegav$. More precisely, say that on her move, after Round $t$ ($\Qfr_{t+1}=\exists$), the reached position in $\SAT(\phi, \BBB)$ is $(\str{L}_t, f_t^\textrm{S})$; in parallel, she reached a position $f_t^\textrm{V}\colon{\Vars}_t \rightarrow A^\Cons$ in $\VER(\phi,\str{A})$, which is such that $(\str{A} \restr{\big(f_t^{\textrm{V}}[{\Vars}_t]{\setminus}\Cons\big)}, f_t^{\textrm{V}})$ agrees with $(\str{L}_t, f_t^\textrm{S})$ on the adequate atoms of $\psi$.
Eloisa looks at the position $f^\textrm{V}_{t+1}\eqdef\omegav(f_t^\textrm{V})$ in $\VER(\phi, \str{A})$ and mimics it in $\SAT(\phi, \BBB)$ by adding to $\str{L}_t$ an element with the same $1$-type as $f^\textrm{V}_{t+1}(y_{t+1})$.
She then extends the structure $\str{L}_t$ to $\str{L}_{t+1}$ by copying the required atoms from $\str{A} \restr{\big(f_{t+1}^{\textrm{V}}[{\Vars}_{t+1}]{\setminus}\Cons\big)}$.
\end{proof}

It is more difficult to show that our satisfaction games are \emph{sound}, that is, if Eloisa has a winning strategy in $\SAT(\phi, \BBB)$, then $\phi$ has a finite model whose outer-types are in $\BBB$. This will be shown in Section~\ref{s:smallmodel}, where we will construct such a model based on Eloisa's
winning strategy. Before this, in the next subsection, we will show that if Eloisa has a winning strategy in the game $\SAT(\phi, \BBB)$ then
she also has one in $\SAT(\phi, \BBB')$, where $\BBB'$ contains only exponentially many $1$-types with respect to the length of $\phi$ (as noticed in Section~\ref{section:background}, in the
presence of constants the number of $1$-types is in general doubly exponential). This additional observation will help us later to
get a tight upper bound on the size of the constructed models.

\subsection{Small number of $1$-types} \label{s:types}

In the following subsection, we assume that Eloisa has a winning strategy $\omegas$ in $\SAT(\phi, \BBB)$. We write $\cF_{\omegas}^\exists$ for $\cF_{\omegas}^{\exists}(\SAT(\phi, \BBB))$.

\medskip\noindent
{\bf Equivalence on $1$-types.}
Let $f_{\mapsto1}\colon\Vars \rightarrow [1]^\Cons$ be the assignment assigning $1$ to every $v \in \Vars$.
If $f\colon {\Vars}_t\rightarrow [k]^\Cons$ is an assignment, with $t\leq M$, then by $f^{\textrm{flat}}:{\Vars}_t \rightarrow [k]^\Cons$ we denote the assignment defined as follows: $f^{\textrm{flat}}(v)=f(v)$ if $f(v) \in \Cons$ and
$f^{\textrm{flat}}(v)=1$ otherwise.

For every assignment $f \in \cF_{\omegas}^\exists$,
we introduce an equivalence relation $\sim_{f}$ on the set $\AAA$ of $1$-types from $\BBB$. Intuitively, $\sim_{f}$ relates 1-types which are equally good for Eloisa when she chooses one for the freshly introduced element
in a position with assignment $f$.

Formally, assuming that $f$ is of order $t$, we set $\alpha_1 \sim_{f} \alpha_2$ if the following two conditions hold:
\begin{enumerate}[label = (\roman*)]
  \item \label{simfitem1} for every atom $\gamma(\bar{v})$ of $\psi$, we have that $\alpha_1, f_{\mapsto1} \models \gamma(\bar{v})$ iff $\alpha_2, f_{\mapsto1} \models \gamma(\bar{v})$;
  \item \label{simfitem2} for every atom $\gamma(\bar{v})$ of $\psi$ such that $y_t \in \bar{v} \subseteq {\Vars}_t$ and $f[\bar{v} \setminus \{y_t\}] \subseteq \Cons$,
	we have that $\alpha_1, f^{\textrm{flat}}(\bar{v}) \models \gamma(\bar{v})$ iff $\alpha_2, f^{\textrm{flat}}(\bar{v}) \models \gamma(\bar{v})$.
\end{enumerate}

We remark that \ref{simfitem1} will be important for atoms $\gamma(\bar{v})$  containing
one variable (and possibly some constants), while \ref{simfitem2} will be used in situations, where $\bar{v}$
has more variables, but only one of them is mapped by $f$ to an unnamed element and the remaining---to some constants. 

It is routine to verify that $\sim_f$ is indeed an equivalence relation over $1$-types.

\medskip\noindent
{\bf New game construction.} We use the relations $\sim_f$ to define from the set $\BBB$ a new set $\BBB'$ of outer-types containing only exponentially many $1$-types.
First, for every $f \in \cF^\exists_{\omegas}$, we fix a choice function $\chc_f\colon \AAA/_{\sim_f} \rightarrow \AAA$ that assigns to every class $[\alpha]_{\sim_f}$ one of its elements.

Then, the set $\BBB'$ consists of the outer-types $\beta'$ for which there exists  an outer-type $\beta\in\BBB$ of the same grade $k$ such that:
\begin{itemize}[nosep]
	\item $\hulltype{\beta'}{1,\ldots, k}=\hulltype{\beta}{1,\ldots, k}$;
	\item for each $i\in[k]$, there exists $f \in \cF_{\omegas}^\exists$  such that
$\type{\beta'}{i} = \chc_{f}([\type{\beta}{i}]_{\sim_{f}})$.
\end{itemize}

Observing that  $\BBB'$ is consistent and closed is routine; it follows from the fact that $\BBB$ is closed and consistent as well.

In effect, the set $\AAAp$ of $1$-types in $\BBB'$ is the set $\{\chc_f([\alpha]_{\sim_{f}}): f \in \cF^\exists_{\omegas}, \alpha \in \AAA\}$.

Let us estimate the size of $\AAAp$. 
For each $f \in \cF_{\omegas}^\exists$, the number of equivalence classes of $\sim_f$ is at most $2^{2\cdot|\atoms(\phi)|}=2^{\cO(|\phi|)}$,
and the number of assignments in $\cF_{\omegas}^\exists$ is at most
$\sum_{t=1}^{\bM}{(K{+}t{+}|\Cons|)^{K+t}} \le \bM \cdot (K{+}\bM{+}|\Cons|)^{K+M},$
which is $2^{\cO(|\phi|\cdot\log|\phi|)}$.
Hence, $|\AAAp|= 2^{\cO(|\phi|\cdot\log|\phi|)}$: the number of $1$-types in $\BBB'$ is exponential in $|\phi|$, as desired.

\medskip\noindent
{\bf The correspondence of plays.}
Now, we show that Eloisa has a winning strategy $\omegas'$ in the game $\SAT(\phi, \BBB')$.

We will construct such a strategy inductively, and additionally, in parallel,
we will construct a partial function $\Gamma$, called \emph{simulation}, from positions in the new game $\SAT(\phi,\BBB')$ to positions in the old game $\SAT(\phi, \BBB)$.
Our simulation $\Gamma$ will be defined for positions in $\SAT(\varphi,\BBB')$ which are reachable when following $\omegas'$, and will always output a position in $\SAT(\varphi,\BBB)$ that is reachable when following $\omegas$. During the play we will keep the invariant that,
for any positions $\rho'$ and $\Gamma(\rho')$ the domains of their structures and their assignments are equal.
At the beginning of our construction, both $\Gamma$ and $\omegas'$ are the empty functions. Let us see how $\SAT(\phi,\BBB')$ evolves.

\smallskip\noindent
{\bf Round 0.}
The move in Round $0$ belongs to Abelard, who plays a position $\rho_0'=(\str{L}_0', f_0')$ for some $k$-outer-type $\str{L}_0'\in\BBB'$.

If $k \le 1$, then $\str{L}_0'$ is  a $0$-type or a $1$-type from $\AAAp \subseteq \AAA$. This means that $\rho_0'$ is a valid position in $\SAT(\phi,\BBB)$, and we can set $\Gamma(\rho'_0)=\rho'_0$.

If $k>1$, then we know from our definition of $\BBB'$ that there exists a $k$-outer-type $\str{L}_0 \in \BBB$, such that $\hulltype{\str{L}_0}{1, \ldots, k}=\hulltype{\str{L}'_0}{1, \ldots, k}$, and, for every $i\in[k]$, we have $\type{\str{L}_0}{i} \sim_{f} \type{\str{L}'_0}{i}$ for some $f \in F_{\omegas}^\exists$. We set $\Gamma(\rho_0')$ to be $(\str{L}_0, f_0')$. In both cases, the domains $L_0$ and $L_0'$ coincide.

\smallskip\noindent
{\bf Round $\mathbf{t{+}1}$.}
Assume that $\omegas'$ and $\Gamma$ are defined for positions of order $t$. We extend them to positions of order $t{+}1$. In both cases below, $\rho'_{t}=(\str{L}'_t,f'_t)$ is a position of order~$t$ that can be reached when following the current $\omegas'$, and $\rho_{t}=(\str{L}_t,f'_t)$ is $\Gamma(\rho_t')$, its image by $\Gamma$.

\emph{Abelard's move.}
If $\Qfr_{t+1}=\forall$, then consider any position $\rho'_{t+1}=(\str{L}'_{t+1},f'_{t+1})$ reached from $\rho'_t$ after Abelard's move.

If Abelard decided not to introduce
a fresh element in this round of $\SAT(\phi,\BBB')$, i.e. $\str{L}'_{t+1}=\str{L}'_t$ and $f'_{t+1}(y_{t+1})\in {L'_t}^\Cons$ then the position $\Gamma(\rho'_{t+1})$ is $(\str{L}_{t},f'_{t+1})$.
If on the other hand
Abelard did introduce a fresh element to obtain $\rho_{t+1}'$, then similarly, an element of the same $1$-type
is added to obtain $\Gamma(\rho_{t+1}')$ (the assignment is also modified accordingly).

\emph{Eloisa's response.}
If  $\Qfr_{t+1}{=}\exists$, then we need to
extend $\omegas$ so that it is defined for positions of order $t{+}1$.
Let $\rho_{t+1}=(\str{L}_{t+1}, f_{t+1})$ be $\omegas(\rho_t)$, i.e. the response suggested by $\omegas$ in the simulation $\SAT(\phi,\BBB)$. Let $a$ be an element newly introduced in $\str{L}_{t+1}$ (i.e. $L_{t+1}=L_t\cup \{a\}$). Then $\omegas'(\rho_{t}')$ is set as $(\str{L}'_{t+1}, f_{t+1})$, where $L'_{t+1}=L'_t \cup \{ a \}=L_{t+1}$, $\str{L}'_{t+1}\restr L'_{t}=\str{L}_t'$,
$\type{\str{L}'_{t+1}}{a}=\chc_{f_{t+1}}([\type{\str{L}_{t+1}}{a}]_{\sim_{f_{t+1}}})$, 
and the hull-types of tuples containing $a$ are obtained by copying
the corresponding hull-types from $\str{L}_{t+1}$. We naturally set $\Gamma(\rho'_{t+1})$ to be $\rho_{t+1}$.

The following claim states some basic invariants of the construction of our simulation $\Gamma$ (Points~\ref{it:corresp_small_i} to~\ref{it:corresp_small_iv}), together with a more crucial property, implying that $\omegas'$ is actually winning for Eloisa (Point~\ref{it:corresp_small_v}).

\begin{claim} \label{c:correspondance1}
Let $\rho_t'=(\str{L}'_t, f'_t)$ be a position of order $t$ that can be reached when following Eloisa's strategy $\omegas'$ defined above, and let $\rho_t=(\str{L}_t, f_t)$ be $\Gamma(\rho_t')$. Then:
\begin{enumerate}[label =(\roman*)]
\item\label{it:corresp_small_i} $L_t'=L_t$, $f_t'=f_t$, and $\str{L}'_t \restr \emptyset = \str{L}_t \restr \emptyset$;
\item\label{it:corresp_small_ii} for universally quantified $v\in\Vars_t$, we have the equivalence $\tp^{\str{L}'_t}\big(f_t(v)\big)\sim_{f} \tp^{\str{L}_t}\big(f_t(v)\big)$,
for some $f \in \cF_{\omegas}^\exists$;
\item\label{it:corresp_small_iii} for existentially quantified $y_i$, $i\in [t]$, we have the equivalence $\tp^{\str{L}'_t}\big(f_t(y_i)\big)\sim_{f^*} \tp^{\str{L}_t}\big(f_t(y_i)\big)$,
  where $f^*$ is the restriction of $f_t$ to ${\Vars}_i$;
\item\label{it:corresp_small_iv} for every tuple of distinct elements $a_1, \ldots, a_i \in L_t$, $2 \le i$,
we have
that  $\hulltype{\str{L}'_t}{a_1, \ldots, a_i}=\hulltype{\str{L}_t}{a_1, \ldots, a_i}$, or both hull-types are undefined;
\item\label{it:corresp_small_v}  for every atom $\gamma(\vs)$ of $\psi$ such that $\vs\subseteq{\Vars}_t$, we have that $\str{L}'_t,f_t\models\gamma(\vs)$ iff $\str{L}_t,f_t\models \gamma(\vs)$.
\end{enumerate}
\end{claim}

Points~\ref{it:corresp_small_i}-\ref{it:corresp_small_iv} follow easily from our constructions of $\omegas'$ and $\Gamma$. We prove the crucial Point~\ref{it:corresp_small_v} in Appendix~\ref{appendix:satgames}.
It implies in particular, by taking $t=M$, that for every $\rho_M'$ of order $M$ reachable when following $\omegas'$, we have the equivalence $\rho_M'\models \psi$ if and only if $\Gamma(\rho_M')\models\psi$. Yet the latter is always true, since $\Gamma(\rho_M')$ can be reached in $\SAT(\varphi,\BBB)$ when following $\omegas$, which is winning in this game. Hence, we can conclude:

\begin{lemma} \label{l:lesstypes}
Assume that $\phi$ is satisfiable. Then there is a consistent and closed set of outer-types $\BBB$, with $2^{O(|\phi|\cdot\log|\phi|)}$ $1$-types, such that Eloisa has a winning strategy in $\SAT(\phi, \BBB)$.
\end{lemma}
 
\section{Small-model construction for $\overline{\text{K}}$} \label{s:smallmodel}
\label{section:models}

In this section, we finally prove 
Theorem~\ref{theorem:dmk-expo-size}.
We fix a satisfiable sentence $\phi$ in $\DMK$ of the shape as in (\ref{eq:prenex}), i.e.~$\varphi$ is~$\forall x_1 \ldots \forall x_\bK.\break\Qfr_{1}y_{1} \ldots \Qfr_\bM y_\bM.~\psi$.
By Lemma~\ref{l:lesstypes}, there exists a set $\BBB$ of outer-types with exponentially many $1$-types such that Eloisa has a winning strategy $\omegas$ in  $\SAT(\phi, \BBB)$.
We write $\cP_{\omegas}^\exists$ for $\cP_{\omegas}^\exists(\SAT(\phi, \BBB))$ and $\cF_{\omegas}^\exists$ for $\cF_{\omegas}^\exists(\SAT(\phi, \BBB))$.

As the foundation for our model $\str{A}$, we will consider an $(\cR,\cQ)$-paradoxical colourful tournament $(\mathcal{T}, \mu, \lambda)$,
where the set $\cQ$ of arc colours is the set of variables from $\phi$,
and the set $\cR$ of vertex colours is the set of positions from $\cP_{\omegas}^\exists$, quotiented by some equivalence relation (in order to get exponentially many such colours).
The vertex set of $\cT$ will be taken for the unnamed domain of $\str{A}$, and we will use the labellings $\mu$ and $\lambda$ to specify interpretations of relational symbols present in the signature $\sigma\eqdef\sigma(\phi)$.

In parallel to the specifications of these interpretations, we will construct a strategy $\omegav$ for Eloisa in $\VER(\phi,\str{A})$. This strategy will be winning, thus ensuring that $\str{A}$ is indeed a model of $\phi$.

Eloisa will simulate a play of $\SAT(\phi, \BBB)$ in parallel to the play of $\VER(\phi,\str{A})$.
The simulation invariants will provide a certain coherency between the positions. Suppose that the current position in the latter game is $f_t^\textrm{V}\colon{\Vars}_t \rightarrow A^\Cons$, then, in the simulation, we will reach a position $\rho_t = (\str{L}_t, f^\textrm{S}_t)$ of the same order. Moreover, for each variable in ${\Vars}_t$, $f_t^\textrm{S}$ and $f_t^\textrm{V}$ will assign either unnamed elements of the same $1$-type, or the same constant. The equality will be preserved (if $f_t^\textrm{V}(v) = f_t^\textrm{V}(v')$ then $f_t^\textrm{S}(v) = f_t^\textrm{S}(v')$ as well), and, finally, the simulation will ensure that for every atom $\gamma(\vs)$ of $\psi$, with $\vs\subseteq{\Vars}_t$, we have $\str{A}, f^\textrm{V}_t\models\gamma(\vs)$ if and only if $\str{L}_t, f^\textrm{S}_t\models\gamma(\vs)$.

Let us give an intuition on how the strategy $\omegav$ is working, when it is Eloisa's turn, for the existentially quantified variable $y_{t+1}$.
Let $\bar{a}=\tuple{a_1, \ldots, a_k}$ be an enumeration of the elements in $f^\textrm{V}_t[{\Vars}_t] \setminus \Cons\subseteq A$,
and let $\bar{v}=\tuple{v_1, \ldots, v_k}$ be a tuple of variables such that  $f_t^\textrm{V}(v_i)=a_i$ (the choice for $\vs$ might not be unique).
Let $\rho_{t+1}	=(\str{L}_{t+1}, f^\textrm{S}_{t+1})$ be the next position in the simulation of $\SAT(\phi,\BBB)$, obtained by following $\omegas$.
Eloisa mimics this move in $\VER(\phi, \str{A})$ by choosing for the variable $y_{t+1}$ an element $b \in A$ that colourfully dominates $\bar{a}$ via the vertex colour representing $\rho_{t+1}$ and~$\bar{v}$. The very importance of the paradoxicality property in our proof comes from this step.
Figure~\ref{fig:game1} depicts a possible choice for Eloisa in the case $t=2$: the arcs and their labels are in blue, while the assignments $f^\textrm{V}_2$ and $f^\textrm{S}_2$ are in red. The color ``representing'' $\rho_{t+1}$ is denoted $\rho$.

\vspace{-1em}
\begin{figure}[h]
  \begin{center}
  \includegraphics[width=\columnwidth]{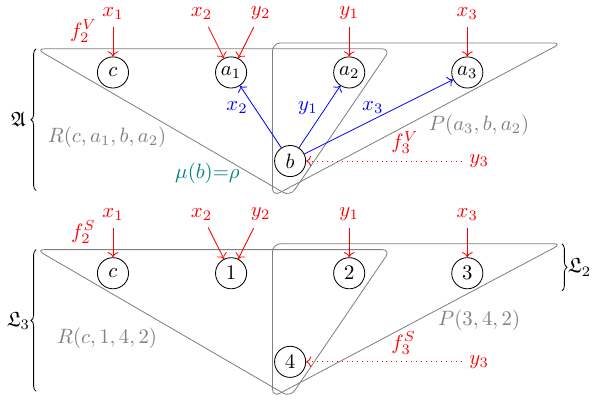}
  \caption{Eloisa's move in Round 3 of $\VER(\phi, \str{A})$ (above),\\ and its simulation in $\SAT(\phi, \BBB)$ (below).} {\label{fig:game1}}
  \end{center}
\end{figure}
\vspace{-1em}

We set the hull-types of all tuples of $\str{A}$ in such a way that the above suggested strategy for Eloisa maintains the similarity between positions in $\VER(\phi,\str{A})$ and their simulations in $\SAT(\phi,\BBB)$.

For example, in Figure~\ref{fig:game1}, if $R(x_1, x_2, y_3, y_1)$ is an atom of $\psi$, and $\str{L}_{3}, f^\textrm{S}_{3} \models R(x_1, x_2, y_3, y_1)$, then $R^{\str{A}}(c, a_1, b, a_2)$ shall be set to true as well. 

\medskip\noindent
{\bf Equivalence of positions.} Our  model construction would work correctly if we simply used the full set $\cP^\exists_{\omegas}$ for the set $\mathcal{R}$ of vertex colours. However, to get a model of the optimal size $2^{\cO(|\phi|\cdot\log|\phi|)}$, we need a smaller set $\mathcal{R}$.\footnote{Notice that a priori the set $\cP^\exists_{\omegas}$ has doubly exponential size, as Abelard's first move in $\SAT(\phi,\BBB)$ is to choose an outer-type from $\BBB$.} This is why we introduce an equivalence relation on $\cP^\exists_{\omegas}$ and select for $\mathcal{R}$ only representatives of its equivalence classes.

Let $(\str{L},f)$ and $(\str{L}',f')$ be positions from $\cP^\exists_{\omegas}$.
We set $(\str{L},f) \sim (\str{L}',f')$ if the following conditions hold:
\begin{enumerate}[label=(\roman*)]
  \item\label{it:eq_i} $L=L'$ and $f=f'$;
  \item\label{it:eq_ii} $\tp^{\str{L}}(f(y_t)) = \tp^{\str{L}'}(f(y_t))$, where $t$ is the order of $f$;
  \item\label{it:eq_iii} for every atom $\gamma(\vs)$ of $\psi$ containing $y_t$, but no $y_i$ with $i>t$, 
	we have $\str{L},f \models \gamma(\vs)$ iff $\str{L}',f \models \gamma(\vs)$.
\end{enumerate}
The definition implies that $\sim$-equivalent positions agree on the atoms of $\psi$ containing the new variable $y_t$.
It is immediate to check that $\sim$ is an equivalence relation over positions from $\cP^\exists_{\omegas}$. Moreover, the number of equivalence classes is at most $|\cF_{\omegas}^\exists|\cdot|\AAA|\cdot2^{|\atoms(\psi)|}$, which is $2^{\cO(|\phi|\cdot\log|\phi|)}$, since, as we observed in Section~\ref{section:satgames}, $\cF_{\omegas}^\exists= 2^{\cO(|\phi|\cdot \log|\phi|)}$ as well.

\medskip\noindent
{\bf Model construction.}
Let $\chc\colon \cP^\exists_{\omegas} {/}_{\sim} \rightarrow \cP^\exists_{\omegas}$ be a choice function selecting a single position from every equivalence class of $\sim$.
We define our set $\cR$ of vertex colours as the image of $\chc$: $\cR=\{ \chc([(\str{L},f)] {/}_{\sim}) : (\str{L},f) \in \cP^\exists_{\omegas} \}$, while our set $\cQ$ of arc colours is the set $\Vars=\xs\cup\ys$. 
Let $\cT = (V,E)$ be an $(\cR, \cQ)$-paradoxical colourful tournament with a pair of labellings $\mu\colon V \rightarrow \cR$ and $\lambda\colon E \rightarrow \cQ$. 
By Lemma~\ref{lemma:colourful-tournaments}, we can assume that $V$ is of size $2^{\cO(|\cQ|\cdot\log|\cQ|)}\times|\cR|\cdot\log|\cR|$, and therefore of size $2^{\cO(|\varphi|\cdot\log|\varphi|)}$. We define the unnamed domain~$A$ of our new model $\str{A}$ to be the vertex set $V$ of $\cT$.

The construction proceeds in the following three stages: \emph{Setting the $1$-types}, \emph{Providing the witnesses} (setting the hull-types for non-singleton tuples of unnamed elements necessary for Eloisa's strategy), and \emph{Completing the structure} (setting the hull-types of the remaining tuples). 
The $1$-types and the hull-types will always be  induced from  $\BBB$. 

We need the following notions. A subset $B \subseteq A$ is \emph{self-dominating} if there exists $b \in B$ such that $b$ dominates $B \setminus \{b\}$ in $\cT$. In this case, assuming that the assignment $f_{t+1}$ in the position $\mu(b)=(\str{L}_{t+1},f_{t+1})$ is of order $t{+}1$ \footnote{This choice of $t{+}1$ rather than $t$ will be convenient next, and it is not problematic since no position in $\cP_{\omegas}^\exists$ can be of order $0$.},
we define the function $g_B\colon B \rightarrow \Vars$ in the following way: $g_B(b)=y_{t+1}$ and $g_B(a)=\lambda(b \rightarrow a)$ for $a \in B \setminus \{ b \}$.
Moreover, we say that $B$ is \emph{properly self-dominating} if $g_B(B) \subseteq {\Vars}_{t+1}$, $f_{t+1} \circ g_B\colon B\to L_{t+1}^\Cons$ is injective, and
$(f_{t+1} \circ g_B)[B] \cap \Cons$ is empty. For example, in Figure~\ref{fig:game1} the sets $\{a_1, a_3, b\}$ and $\{a_1, a_2, a_3, b\}$ are properly self-dominating.

\smallskip\noindent
\emph{Stage 1. Assigning the 1-types.} For each element $a \in A$, we do the following.
Let $(\str{L}_t,f_t)=\mu(a)$ be the colour of $a$, of order $t$, and let us set  $\tp^{\str{A}}(a)$ as $\tp^{\str{L}_t}(f_t(y_t))$. Notice that this step sets also the truth values of all the ground facts (i.e. facts on constants)
and that this is done without conflicts at the level of $0$-types, as all the $1$-types we assign to elements are from $\BBB$, which is consistent.

\smallskip\noindent
\emph{Stage 2. Providing the witnesses.} 
For each properly self-dominating subset $B \subseteq A$ of cardinality $k$ between 2 and $\bK{+}\bM$, we do the following.
Let $b\in B$ be the element dominating $B \setminus \{b\}$. Let $(\str{L}_{t+1},f_{t+1})=\mu(b)$, of order $t{+}1$.
If there is a position $(\str{L}^*_{t+1},f_{t+1}) \sim (\str{L}_{t+1},f_{t+1})$ satisfying $\type{\str{A}}{a}=\type{\str{L}^*_{t+1}}{f_{t+1}(g_B(a))}$ for each $a \in B$, then we take an enumeration $\bar{a}=\tuple{a_1, \ldots, a_k}$ of the elements of $B$ and set $\hulltype{\str{A}}{\bar{a}}$ as $\hulltype{\str{L}^*}{f_{t+1}(g_B(\bar{a}))}$; otherwise we leave this hull-type undefined.

Notice that when defining the hull-types of two distinct non-empty subsets, no conflicts can arise as they
 share no atoms. 

\smallskip\noindent
\emph{Stage 3. Completing the structure.}
For each subset $B \subseteq A$ of cardinality $k$ between 2 and $\bK{+}\bM$ with a hull-type still undefined, we do the following.
Take any enumeration  $\bar{a}=\tuple{a_1, \ldots, a_k}$ of the elements of $B$ and select a $k$-outer-type $\beta \in \BBB$ such that $\type{\str{A}}{a_i} = \type{\beta}{i}$ for each $i\in[k]$.
We are guaranteed that such $\beta$ exists by the closedness of $\BBB$.
We then set $\hulltype{\str{A}}{\bar{a}}$ as $\hulltype{\beta}{1, \ldots, k}$.
By the same argument as before, we do not introduce any conflicts in this stage too.

The remaining still unspecified facts in $\str{A}$ contain more than  $K{+}M$ unnamed elements, and since this number is bigger than the
number of variables in $\phi$, they are irrelevant for the truth value of $\phi$. Hence, we can set all of them, e.g., to be false.
This finishes the construction of $\str{A}$.

\medskip

To understand some subtleties of the construction described above, let us look again at Figure~\ref{fig:game1} and
see how $b$ would be prepared as a correct choice for Eloisa for the position $f^\textrm{V}_{3}$ if $\psi$ 
contains, say, the atoms $R(x_1, x_2, y_3, y_1)$ and $P(x_3, y_3, y_1)$.  These two atoms, under $f_{3}^\textrm{V}$, become
respectively $R(c, a_1, b, a_2)$ and $P(a_3, b, a_2)$.
Hence, the definitions of $\hulltype{\str{A}}{a_1, a_2, b}$ and
$\hulltype{\str{A}}{a_2, a_3, b}$ are of high importance. They are defined in two different steps of Stage 2, and it might be that these hull-types may be taken from different positions, say $(\str{L}^*_{3}, f^{\textrm{S}}_{3})$ and $(\str{L}^{**}_{3}, f^{\textrm{S}}_{3})$. Nevertheless, 
it is not dangerous. Indeed, the construction makes sure that they both belong to the equivalence class of $(\str{L}_{3}, f^{\textrm{S}}_{3})$, and hence by Condition~\ref{it:eq_iii} of the definition
of $\sim$, for every atom $\gamma(\vs)$ of $\psi$
containing $y_3$ as its maximal variable, we have that $\str{L}^*_{3}, f_{3} \models \gamma$  iff $\str{L}_{3}, f_{3} \models \gamma$ iff  $\str{L}^{**}_{3}, f_{3} \models \gamma$. Hence, the truth-values of all the important atoms are as promised by the vertex colour $\rho$ assigned to $b$.

\medskip\noindent
{\bf Correspondence of plays.}
We show now that $\str{A} \models \phi$. To this end we prove the existence of a winning  strategy $\omegav$ for
Eloisa in  $\VER(\phi,\str{A})$. Similarly to Section \ref{s:types}, in the construction of $\omegav$, we will simulate a play of $\SAT(\phi, \BBB)$ 
and base Eloisa's response in $\VER(\phi,\str{A})$ on her winning strategy $\omegas$ in $\SAT(\phi, \BBB)$.
We simultaneously define $\omegav$, and a mapping $\Gamma_{\textrm{V}\to \textrm{S}}$ from positions of $\VER(\phi,\str{A})$ reachable when following $\omegav$ to positions of $\SAT(\phi, \BBB)$ reachable when following $\omegas$. The construction is done by induction on the order of the position $f^\textrm{V}\colon{\Vars}_t\to A^\Cons$.
Initially $\Gamma_{\textrm{V}\to \textrm{S}}$ and $\omegav$ are empty functions.

\smallskip\noindent
{\bf Round 0.}
Let us start with a position $f^\textrm{V}_0\colon{\Vars}_0 \rightarrow A^\Cons$, chosen by Abelard in Round $0$ of $\VER(\phi,\BBB)$.
Let $\bar{a}=\tuple{a_1, \ldots, a_k}$ be an enumeration of
the elements of $f^\textrm{V}_0[{\Vars}_0] \setminus \Cons$. Let $\str{L}_0=\outertype{\str{A}}{\bar{a}}$ and let $f^\textrm{S}_0\colon{\Vars}_0 \rightarrow L_0^\Cons$ be defined as follows:
if $f^\textrm{V}_0(x_i) \in \Cons$, then $f^\textrm{S}_0(x_i)=f^\textrm{V}_0(x_i)$; and if $f^\textrm{V}_0(x_i)=a_j$, for some $j$, then $f^\textrm{S}_0(x_i)=j$.  We set $\Gamma_{\textrm{V}\to\textrm{S}}(f^\textrm{V}_0)$ to be $(\str{L}_0, f_0^\textrm{S})$.

\smallskip\noindent
{\bf Round $\mathbf{t{+}1}$.}
Let us assume now that $\omegav$ and $\Gamma_{\textrm{V}\to\textrm{S}}$ are defined for positions of order $t$, and let us define them for positions of order $t{+}1$. In the following, $f^\textrm{V}_{t}\colon{\Vars}_t\rightarrow A^\Cons$ is a position of order $t$ reachable when following the current $\omegav$.

\emph{Abelard's move.}
If $\Qfr_{t+1}{=}\forall$, consider any position $f^\textrm{V}_{t+1}\colon\break{\Vars}_{t+1}\rightarrow A^\Cons$, chosen by Abelard and extending $f^\textrm{V}_t$ in an appropriate way. Let $(\str{L}_{t}, f^\textrm{S}_t)$ be the position $\Gamma_{\textrm{V}\to\textrm{S}}(f^\textrm{V}_{t})$. We set $\Gamma_{\textrm{V}\to\textrm{S}}(f^\textrm{V}_{t+1})$ 
 to be the position $(\str{L}_{t+1}, f^\textrm{S}_{t+1})$ defined as follows. If, for  $f^\textrm{V}_{t+1}(y_{t+1})$, Abelard chose either a constant $c$ or an element already chosen before as $f^\textrm{V}_t(v)$, for some $v\in\Vars_t$, then $\str{L}_{t+1}$ is $\str{L}_t$, and $f^\textrm{S}_{t+1}=f^\textrm{S}_t \cup \{y_{t+1}\mapsto c~/~f^\textrm{V}_t(v)\}$ ($c~/~f^\textrm{V}_t(v)$ depending on the subcase). If on the other hand, he chose an element $a$ not already chosen before, then we obtain $\str{L}_{t+1}$ by extending $\str{L}_t$
with a fresh element and setting its $1$-type to $\type{\str{A}}{a}$; finally $f^\textrm{S}_{t+1}$ extends $f^\textrm{S}_t$ by assigning $y_{t+1}$ to this new element.

\emph{Eloisa's response.}
If $\Qfr_{t+1}{=}\exists$, then let $\bar{a}=\tuple{a_1, \ldots, a_k}$ be an enumeration of  the elements in $f^\textrm{V}_{t}[{\Vars}_t] \setminus \Cons\subseteq A$,
and let $\bar{v}=\tuple{v_1, \ldots, v_k}$ be a tuple of variables in ${\Vars}_t$ such that  $f^\textrm{V}_{t}(v_i)=a_i$ (there may be more than one choice for $\vs$). Let $\rho_t=(\str{L}_t, f^\textrm{S}_t)$ be $\Gamma_{\textrm{V}\to\textrm{S}}(f^\textrm{V}_{t})$, and let $\rho_{t+1}=(\str{L}_{t+1}, f^\textrm{S}_{t+1})$ be $\omegas(\rho_{t})$, i.e. the corresponding move suggested by $\omegas$ for Eloisa in $\SAT(\phi, \BBB)$.
To find a simulating move in $\VER(\phi, \str{A})$, let us take the vertex colour $\rho \in \mathcal{R}$ $\sim$-equivalent to  $\rho_{t+1}$,
and let $b$ be an element of $A$ that colourfully dominates $\bar{a}$ via $\rho$ and~$\bar{v}$. We set  $\omegav(f^\textrm{V}_t)$ to be the assignment $f^\textrm{V}_{t+1}=f^\textrm{V}_t \cup \{y_{t+1}\mapsto b\}$. Finally, we  define $\Gamma_{\textrm{V}\to\textrm{S}}(f^\textrm{V}_{t+1})$ to be $\rho_{t+1}$.

\medskip

This is the end of our construction, which ensures some important invariants, stated in the following claim:

\begin{claim} \label{c:correspondance2}
  Let $f^\textrm{V}_{t}\colon{\Vars}_t\to A^\Cons$ be a position of $\VER(\phi,\str{A})$, reachable when following $\omegav$, and let $(\str{L}_t, f^\textrm{S}_t)$ be $\Gamma_{\textrm{V}\to\textrm{S}}(f^\textrm{V}_t)$, position of $\SAT(\phi, \BBB)$ reachable when following $\omegas$. Then:
\begin{enumerate}[label=(\roman*)]
\item \label{it:corresp_pos_i} $f_t$ and $f'_t$ are both of the same order (here $t$);
\item \label{it:corresp_pos_ii} for every variables $v_1,v_2$ and $v$ in $\Vars_t$: $f^\textrm{V}_t(v_1)=f^\textrm{V}_t(v_2)$ iff $f_t^\textrm{S}(v_1)=f_t^\textrm{S}(v_2)$; $f_t^\textrm{V}(v)=c$  iff $f_t^\textrm{S}(v)=c$  for every constant $c$; if $f^{\textrm{V}}_t(v)$ is unnamed then $\type{\str{A}}{f_t^\textrm{V}(v)}=\type{\str{L}_t}{f^\textrm{S}_t(v)}$;
\item \label{it:corresp_pos_iii} for every atom $\gamma(\vs)$ of $\psi$ such that $\vs\subseteq{\Vars}_t$, we have that $\str{A}, f^\textrm{V}_t\models \gamma(\vs)$ iff $\str{L}_{t}, f^\textrm{S}_t\models \gamma(\vs)$.
\end{enumerate}
\end{claim}

\begin{proof}
Points~\ref{it:corresp_pos_i} and~\ref{it:corresp_pos_ii} directly follow from the construction.
We prove Point~\ref{it:corresp_pos_iii} by induction on $t$.

If $t=0$, then, by~\ref{it:corresp_pos_i}, the domain of the two assignments is ${\Vars}_0$, which means that $\bar{v}$ either contains all of the $x_i$ or at most one of them.
In our construction, we have set $\str{L}_0$ to be the outer-type of some enumeration of all the elements of $f^\textrm{V}[{\Vars}_0] \setminus \Cons$, and set $f^\textrm{V}$ \emph{in accordance} with $f^\textrm{S}$. As the outer-types define the truth-values of the atoms containing all of their elements or at most one of them, the equivalence follows.

Assume now that the claim holds for $t$, consider an assignment $f^\textrm{V}_{t+1}\colon {\Vars}_{t+1}\to A^\Cons$ extending $f^\textrm{V}_t$,
and write $(\str{L}_{t+1}, f^\textrm{S}_{t+1})$ for $\Gamma_{\textrm{V}\to\textrm{S}}(f^\textrm{V}_{t+1})$.  
By the inductive
assumption, $(\str{A}, f^\textrm{V}_{t+1})$ and $(\str{L}_{t+1}, f^\textrm{S}_{t+1})$ agree on the atoms of $\psi$ whose variables are contained in ${\Vars}_{t}$. We need to consider
the atoms $\gamma(\bar{v})$ whose variables $\vs$ contain $y_{t+1}$.  

If $\Qfr_{t+1}=\forall$, then, since $y_{t+1}$ is not special,~$\vs$ may contain no other variables, and we simply use~\ref{it:corresp_pos_ii}.

If $\Qfr_{t{+}1}=\exists$, then, assuming $L_{t}=[k]$, we have $f^\textrm{S}_{t+1}(y_{t+1})=k{+}1$. We define $b=f^\textrm{V}(y_{t+1})$
and recall that the vertex colour $\rho$ of $b$ satisfies $\rho \sim (\str{L}_{t+1}, f^\textrm{S}_{t+1})$.
Take an enumeration $b, a_1, \ldots, a_\ell \in A$ of the elements of $f^\textrm{V}_{t+1}[\bar{v}] \setminus \Cons$
and the corresponding enumeration $k{+}1, a'_1, \ldots, a'_\ell \in L_{t+1}$ of $f^\textrm{S}_{t+1}[\bar{v}] \setminus \Cons$. 
Our construction sets the hull-type of $\langle b, a_1, \ldots, a_\ell\rangle$ to be the hull-type of
$\langle k{+}1, a'_1, \ldots, a'_\ell\rangle$ in some structure
$\str{L}^*_{t+1}$ for which $(\str{L}^*_{t+1}, f^{\textrm{S}}_{t+1}) \sim \rho \sim (\str{L}_{t+1}, f^{\textrm{S}}_{t+1})$.
By Condition~\ref{it:eq_iii} of the definition of $\sim$, we have the desired equivalence.
\end{proof}

Again, the fact that $\omegav$ is winning in $\VER(\phi,\str{A})$ comes from the last point of the claim, as $\omegas$ is winning in $\SAT(\phi,\BBB)$. This allows us to conclude Theorem~\ref{theorem:dmk-expo-size}.
In the end, we remark that our construction can be made fully deterministic basing on the explicit construction of paradoxical colourful tournaments (see Appendix~\ref{appendix:graphs}).

\section{Conjunctions of $\overline{\text{K}}$-sentences}
\label{section:conjunctions}

The upper bound on the size of models of satisfiable formulas we get for \DMK{} in Theorem~\ref{theorem:dmk-expo-size} can be transferred to \DMDK{}. The adaptations are somewhat technically complex, but rather routine and hence we only shortly outline them here.

Consider a sentence $\phi=\bigwedge_{i=1}^m \phi_i$ in $\DMDK$. All the sentences $\phi_i$'s are as in (\ref{eq:prenex}) (with possibly different $K$'s and $M$'s), but with the variables renamed so that
the sets of variables of the different $\phi_i$ are disjoint.

We extend our games from Section \ref{section:satgames}  by a pre-initial Round (-1) in which Abelard chooses one sentence $\phi_i$. Then the game proceeds as previously on the chosen $\phi_i$. 
This way, the players construct assignments
whose domains contain only variables of this $\phi_i$.

The new game $\VER^+(\phi, \str{A})$ is still, essentially, the standard verification game for First-Order Logic, and it is still very standard that Eloisa has a winning strategy iff $\str{A} \models \phi$. As for the new game $\SAT^+(\phi, \BBB)$, it corresponds to checking satisfiability of each $\phi_i$ independently, but over a
common set of outer-types $\BBB$.

Lemma~\ref{l:satandgames} remains true, with a proof almost identical: we start with any model $\str{A}_0 \models \phi$, augment it as previously, in order to obtain another model
$\str{A} \models \phi$, in which Eloisa has a proper winning strategy in the game $\VER^+(\phi, \str{A})$.
As $\BBB$ we take the set of all outer-types realised in $\str{A}$,
of grade smaller or equal to the maximal number of variables in the $\phi_i$'s. Consider a play of $\SAT^+(\phi, \BBB)$. It is opened by Abelard, who chooses a particular $\phi_i$. We consider the same choice in the game $\VER^+(\phi, \str{A})$. Now, in order to win ${\SAT}^+(\phi, \BBB)$  Eloisa just tries to win ${\SAT}({\phi_i}, \BBB)$,
basing her moves on her winning proper strategy in $\VER(\phi_i, \str{A})$, as in the proof of Lemma \ref{l:satandgames}.

Lemma \ref{l:lesstypes} also remains true for $\SAT^+(\phi,\BBB)$. Again we define similar equivalence relations $\sim_f$. As we assume that different $\phi_i$'s have disjoint sets
of variables, every $f$ may be used for one particular $\phi_i$ only; however, in Condition~\ref{simfitem1} of the definition of $\sim_f$, we use the set of all atoms of $\phi$, rather than
just those of $\phi_i$. As previously, $\BBB'$ is obtained by selecting a single representative from each equivalence class of every equivalence relation. 
It is not difficult to show that still $|\AAAp|=2^{\mathcal{O}(|\phi|\cdot \log|\phi|)}$ (an additional factor that needs to be considered is the number of
conjuncts of $\phi$, but it is only linear in $|\phi|$). To win the new game $\SAT^+(\phi,\BBB')$ Eloisa
simulates a play of $\SAT^+(\phi, \BBB)$ as in the proof of Lemma \ref{l:lesstypes}: on her moves she uses the equivalence relations defined for assignments corresponding to $\phi_i$ chosen in the initial round by Abelard.

The required changes in the model construction (Section \ref{section:models}) and in its correctness proof are similar in spirit.
This time, in Condition~\ref{it:eq_iii} of the definition of the relation $\sim$ on positions, we can restrict attention to atoms of $\phi_i$ corresponding
to the assignment $f$ in the considered positions (recall that the domain of $f$ determines $\phi_i$).  
We note that, in a sense, every $\phi_i$ has now its own vertex colours, but since the number of the $\phi_i$'s is at most linear in $|\phi|$, the whole size of the model is still
$2^{\mathcal{O}(|\phi|\cdot\log|\phi|)}$.
In effect we get:

\begin{theorem}  If $\phi$ is a satisfiable sentence in $\DMDK$, then it admits a model of size $2^{\cO(|\phi|\cdot\log|\phi|)}$. Hence, the satisfiability problem of $\DMDK$ is $\NExpTime$-complete.
\end{theorem}

It is easy to show that, if we can solve the satisfiability problem for $\DMDK$, then we can also solve it for any positive Boolean combination of sentences in $\DMK$: it suffices to  non-deterministically guess  which of the sentences will evaluate to true, check that the
guess indeed guarantees satisfaction of the whole formula, and then solve the problem for the conjunction of the guessed sentences. 
\section{A family of tight examples}
\label{section:lower-bounds}

We prove in this section that the upper bound $2^{\cO(|\phi|\cdot \log |\phi|)}$, which we got for the finite model property of $\DMDK$, is optimal, and already reached in the fragment $\Skolem$:

\begin{restatable}{proposition}{phinconstantfree}
	\label{proposition:phi-n-constantfree}
There exists a family $(\varphi_n)_{n\in\N}$ of satisfiable formulas in $\DMK$ such that each $\varphi_n$ has size $\cO(n)$ and enforce models of at least $2^{\Omega(n\cdot\log n)}$ elements.
	
	Moreover, the formula $\varphi_n$ can be assumed to be in the fragment $\Skolem$, with a unique existential quantifier, and without constant symbols.
\end{restatable}

For simplicity, we propose here formulas with constant symbols. The reader will find details for a constant-free variant in Appendix~\ref{appendix:constant-free-lowerbound}.

Let us construct this formula $\varphi_n$, for some natural $n\geq 3$. The idea is to enforce, for every permutation $\pi$ of $[n]=\{1,2,\ldots, n\}$, the existence of a witness $w_\pi$, and to ensure that any two distinct permutations $\pi\neq\pi'$ have distinct witnesses $w_{\pi}\neq w_{\pi'}$. This way, we will be sure that any model of the formula will have at least $n!=2^{\Omega(n\cdot\log n)}$ elements.

The formula $\varphi_n$ is of the shape \[\forall x_1, \ldots, x_n, y, r_1, \ldots, r_{n-2}.~\exists w.~\psi_n.\]

The variables $x_i$'s, $y$, and $r_i$'s, which are quantified universally, are the special variables of $\varphi_n$, hence $\varphi_n$ is of grade $2n{-}1$. The variable $w$ is the only one quantified existentially. The formula $\psi_n$ is the conjunction of the formulas defined in the following paragraphs.

The signature of the formula consists of the set $\Cons_n=\{c_1,\ldots, c_n, q_0, q_1\}$ of constant symbols, and of the set $\Rels_n =\{P,W,C_\rightarrow, C_\leftarrow,S, Z\}$, where all the symbols are of arity $2n{+}1$. For the sake of readability, we will divide the arguments of these relations into three parts of lengths $n$, one, and $n$ respectively, and will therefore write atoms such as $P(x_1,\ldots, x_n\mid y\mid q_0, r_1, \ldots, r_{n-2}, q_1)$.

The first point is to generate every permutation of $[n]$ with the relation $P$: for every $\pi\in\Perms_n$ (the set of permutations of $[n]$), we shall have $P(c_{\pi(1)}, \ldots, c_{\pi(n)}\mid q_0\mid q_0, \ldots, q_0)$ satisfied. As it is known that any permutation of $[n]$ can be generated from the identity via the permutation that switches $1$ and $2$ and the cyclic permutation $\gamma\colon1\mapsto 2\mapsto\cdots\mapsto n\mapsto 1$, we can achieve this with the formula $\mu_{n,\textrm{perm}}$ defined as
\begin{flalign*}
	&P(c_1,\ldots, c_{n}\mid q_0\mid q_0,\ldots, q_0)\\
	\wedge
	\big[&P(x_1,\ldots,x_{n}\mid y\mid r_1, \ldots, r_{n-2}, q_0, q_0) \rightarrow\\
	&\big(P(x_2, x_1, x_3,\ldots, x_n\mid y\mid r_1, \ldots, r_{n-2}, q_0, q_0)\\
	&\wedge P(x_2,\ldots, x_n, x_1\mid y\mid r_1, \ldots, r_{n-2}, q_0, q_0)\big)\Big].
\end{flalign*}

While it is true that the relation $P$ focusses on its first $n$ variables, we require it to mention $y$ and the $r_i$'s as well, in order for the atoms to contain all the special variables.

Then, we ensure the existence of the witness $w$ via the following formula $\mu_{n,\textrm{witness}}$:
\begin{flalign*}
	&P(x_1,\ldots,x_{n}\mid y\mid r_1, \ldots, r_{n-2}, q_0, q_0) \rightarrow\\
	&W(x_1,\ldots, x_n\mid w\mid q_0, \ldots, q_0).
\end{flalign*}

Now, we have ensured the existence of a witness $w_\pi$ for every permutation $\pi\in\Perms_n$, which satisfies $P(c_{\pi(1)},\dots, c_{\pi(n)}\mid w_\pi\mid q_0,\ldots, q_0)$. Now, to make sure that all these witnesses are distinct, we need to forbid $P(c_{\pi'(1)},\dots, c_{\pi'(n)}\mid w_\pi\mid q_0,\ldots, q_0)$, for any permutation $\pi'\neq\pi$.

For this, we will use the following lemma, proposing a certain decomposition for permutations distinct from $\pi$:

\begin{lemma}
	\label{lm:perm_distinct}
	Let $\pi,\pi'\in\Perms_n$. The permutation $\pi'$ is distinct from $\pi$ if and only if it can be decomposed as $\gamma^{-j}\circ \rho\circ \gamma^k\circ\pi$, where: $\gamma$ is the cyclic permutation $1\mapsto 2\mapsto\cdots\mapsto n\mapsto 1$; $0\leq j < k < n$; and $\rho$ is a permutation of $[n]$ satisfying $\rho(n)=n$.
\end{lemma}

\begin{proof}
	First, suppose $\pi'\neq \pi$. Necessarily, there exists an index $i\in[n]$ such that $\pi'(i)>\pi(i)$. We set $k$ as~$n{-}\pi(i)$, and~$j$ as~$n{-}\pi'(i)$, and we get $0\leq j<k< n$. If now we define~$\rho$ as $\gamma^{j}\circ \pi'\circ\pi^{-1}\circ \gamma^{-k}$, then by definition we obtain the equality~$\gamma^{-j}\circ \rho\circ\gamma^k\circ \pi=\pi'$. It remains to verify that~$\rho$ assigns the index $n$ to itself:
		\begin{flalign*}
\rho(n)&=\gamma^{j}\circ\pi'\circ \pi^{-1}\circ \gamma^{-k}(n)
				=\gamma^j\circ\pi'\circ \pi^{-1}\circ \gamma^{\pi(i)}(n)\\
                &=\gamma^j\circ\pi'\circ \pi^{-1}\big(\pi(i)\big)
				=\gamma^j\circ\pi'(i)
				=\gamma^{-\pi'(i)}\big(\pi'(i)\big)
				=n.
\end{flalign*}

	Hence, $\pi'\neq\pi$ can indeed be decomposed in the desired form.
	
	Now, we show that $\pi$ itself cannot. Suppose that $\pi=\gamma^{-j}\circ \rho\circ \gamma^k\circ\pi$, with the conditions of the statement holding. Then it means that the identity permutation is equal to $\gamma^{-j}\circ \rho\circ \gamma^k$. Yet it can be checked that the latter permutation assigns the index $n{-}k$ to the index~$n{-}j\neq n{-}k$, and hence it cannot be the identity.
\end{proof}

This lemma will help us generating every $\pi'\neq\pi$ from the original permutation $\pi$. We enumerate all of these $\pi'$'s using the relations $C_\rightarrow$, 
$S$, and $C_{\leftarrow}$, as follows: with the use of $C_\rightarrow$, we generate all the permutations of the shape $\gamma^k\circ\pi$; with the use of $S$, we generate the permutations of the shape $\rho\circ \gamma^k\circ\pi$, with $\rho$ assigning $n$ to itself; and finally, with the use of $C_\leftarrow$, we generate the permutations of the shape $\gamma^{-j}\circ \rho \circ \gamma^k\circ\pi$, with $j<k$, i.e. all the permutations distinct from $\pi$.

The first $n$ arguments in these relations will indicate the current considered permutation. On the other side, the last $n$ arguments will be seen as a counter: it will consist of a certain number of occurrences of the constant $q_1$, followed by a certain number of occurrences of the constant $q_0$. The number of occurrences of the constant $q_1$ indicating, in unary, the difference between $k$ and $j$. It is important to make sure that the counter never reaches $0$ nor $n$, so that $\pi'$ can never equal $\pi$.

First, the following formula $\mu_{n,\textrm{cyclic}}$ generates the permutations of the form $\gamma^k\circ \pi$, with $1\leq k< n$:
\begin{flalign*}
	\big[&W(x_1,\ldots,x_{n}\mid y\mid r_1, \ldots, r_{n-2}, q_0, q_0) \rightarrow\\
	&C_\rightarrow(x_2,\ldots, x_n, x_1\mid y\mid q_1, r_1, \ldots, r_{n-2}, q_0)\big]\\
	\wedge\big[&C_\rightarrow(x_1,\ldots,x_{n}\mid y\mid r_1, \ldots, r_{n-2}, q_0, q_0) \rightarrow\\
	&C_\rightarrow(x_2,\ldots, x_n, x_1\mid y\mid q_1, r_1, \ldots, r_{n-2}, q_0)\big].
\end{flalign*}

Notice that, while the counter is incremented, we never allow it to reach $n$, as the last argument always remains $q_0$.

Now, we can generate the permutations of the shape $\rho\circ\gamma^k\circ\pi$, where $\rho$ assigns $n$ to itself. Such $\rho$'s are nothing else than permutations of the smaller set $[n{-}1]=\{1,2,\ldots, n{-}1\}$, and can be obtained via the permutation switching $1$ and $2$ and the pseudo-cyclic permutation $1\mapsto 2\mapsto\cdots\mapsto n{-}1\mapsto 1$. This justifies the following formula $\mu_{n,\textrm{smaller}}$, similar to the formula $\mu_{n,\textrm{perm}}$:
\begin{flalign*}
	\big[&C_\rightarrow(x_1,\ldots,x_{n}\mid y\mid q_1, r_1, \ldots, r_{n-2}, q_0) \rightarrow\\
	&S(x_1,\ldots, x_n\mid y\mid q_1, r_1, \ldots, r_{n-2}, q_0)\big]\\
	\wedge\Big[&S(x_1,\ldots, x_n\mid y\mid q_1, r_1, \ldots, r_{n-2}, q_0) \rightarrow\\
	&\big(S(x_2, x_1, x_3,\ldots, x_n\mid y\mid q_1, r_1, \ldots, r_{n-2}, q_0)\\
	&\wedge S(x_2,\ldots, x_{n-1}, x_1, x_n\mid y\mid q_1, r_1, \ldots, r_{n-2}, q_0)\big)\Big].
\end{flalign*}

Finally, it remains to apply the inverse of the permutation $\gamma$. We do it via the formula $\mu_{n,\textrm{cylic}^{-1}}$: it can apply $\gamma^{-1}$ as long as there remains more than one $q_1$ in the counter:
\begin{flalign*}
	\big[&S(x_1,\ldots,x_{n}\mid y\mid q_1, r_1, \ldots, r_{n-2}, q_0) \rightarrow\\
	&C_{\leftarrow}(x_1,\ldots, x_n\mid y\mid q_1, r_1, \ldots, r_{n-2}, q_0)\big]\\
	\wedge\big[&C_\leftarrow(x_1,\ldots,x_{n}\mid y\mid q_1, q_1, r_1, \ldots, r_{n-2}) \rightarrow\\
	&C_\leftarrow(x_n, x_1,\ldots, x_{n-1}\mid y\mid q_1, r_1, \ldots, r_{n-2}, q_0)\big].
\end{flalign*}

Let $\pi\in\Perms_{n}$. Since we have $W(c_{\pi(1)},\ldots, c_{\pi(n)}\mid w_\pi\mid q_0,\ldots, q_0)$,  Lemma~\ref{lm:perm_distinct}, together with the formulas $\mu_{n,\textrm{cyclic}}$, $\mu_{n,\textrm{smaller}}$, and $\mu_{n,\textrm{cylic}^{-1}}$, ensures that for any $\pi'\in\Perms_{n}$, an atom of the shape $C_\leftarrow(c_{\pi'(1)},\ldots, c_{\pi'(n)}\mid w_\pi\mid q_1,\ldots, q_1, q_0,\ldots, q_0)$ is satisfied if and only if $\pi'\neq\pi$, and in this case the number of $q_1$'s occurring in the counter is the difference $k{-}j$, these two numbers being obtained as in the lemma.

Once we have this, we can make sure that $W(c_{\pi'(1)},\ldots,\break c_{\pi'(n)}\mid w_\pi\mid q_0,\ldots, q_0)$ does not hold for $\pi'\neq\pi$. This will be enough to conclude. We do it via the formula $\mu_{n,\textrm{decrease}}$: starting from the atom of the previous paragraph, and using the predicate $Z$, it decreases the counter, until the latter reaches $0$:
\begin{flalign*}
	\big[&C_\leftarrow(x_1,\ldots,x_{n}\mid y\mid q_1, r_1, \ldots, r_{n-2}, q_0) \rightarrow\\
	&Z(x_1,\ldots, x_n\mid y\mid r_1, \ldots, r_{n-2}, q_0, q_0)\big]\\
	\wedge\big[&Z(x_1,\ldots,x_{n}\mid y\mid q_1, r_1,\ldots, r_{n-2}, q_0) \rightarrow\\
	&Z(x_1,\ldots, x_n\mid y\mid r_1, \ldots, r_{n-2}, q_0, q_0)\big].
\end{flalign*}

Finally, we can negate the relation $W$, in the formula $\mu_{n,\textrm{neg}}$:
\begin{flalign*}
	&Z(x_1,\ldots,x_{n}\mid y\mid r_1, \ldots, r_{n-2}, q_0, q_0) \rightarrow\\
	&\neg W(x_1, \ldots, x_n\mid y\mid  r_1, \ldots, r_{n-2}, q_0, q_0).
\end{flalign*}

As introduced above, we define $\psi_n$ as the conjunction of
$\mu_{n,\textrm{perm}}$, $\mu_{n,\textrm{witness}}$, $\mu_{n,\textrm{cyclic}}$, $\mu_{n,\textrm{smaller}}$, $\mu_{n,\textrm{cylic}^{-1}}$, $\mu_{n,\textrm{decrease}}$ and $\mu_{n,\textrm{neg}}$.
The final formula $\varphi_n$ defined as $\forall\bar{x}, y, \bar{r}.~\exists w.~\psi_n$ is as desired, as expressed in the following two claims
proved in Appendix~\ref{appendix:constant-free-lowerbound}.

\begin{restatable}{claim}{phinsatisfiable}
	\label{claim:phi-n-satisfiable}
	The so-defined formula $\varphi_n$ is in $\Skolem$, of size linear in $n$, and is satisfiable.
\end{restatable}

\begin{restatable}{claim}{phinlowerbound}
	\label{claim:phi-n-lowerbound}
	Any model of $\varphi_n$ has at least $2^{\Omega(n\cdot\log n)}$ elements.
\end{restatable}

\section{Parametrised variants of $\overline{\text{K}}$}
\label{section:skolem}

In Section~\ref{section:models}, we established that any satisfiable formula in $\DMK$ has a model of size $2^{\cO(|\phi|\cdot\log |\phi|)}$, and,
in Section~\ref{section:lower-bounds}, we provided, for every $n{\ge}3$, a formula $\varphi_n$ enforcing models of this size. 

This formula $\varphi_n$ has length linear in $n$ and admits a unique existential quantifier. Yet, a number of universal ones is unbounded.

On the other side, two important subclasses of $\DMK$ consist of formulas with a bounded number of universal quantifiers and are known for admitting smaller models of size $2^{\cO(|\phi|)}$: the Ackermann class (the prefix-class $\forall\exists^*$) and the G\"odel class (the prefix-class $\forall\forall\exists^*$).

Considering these observations, it is natural to ask which role exactly play the universal quantifiers in the gap between models of size $2^{\cO(|\phi|)}$ and those of size $2^{\cO(|\phi|\cdot\log |\phi|)}$.

More precisely, let us define, for a fixed $k\in\N$, the class $\ParamDMK{k}$ consisting of formulas in $\DMK$ that have at most $k$ universal quantifiers.
(We remark that both the Ackermann class and the G\"odel class are contained already in $\ParamDMK{2}$.)

We answer our doubts by showing that every individual class $\ParamDMK{k}$ admits a $2^{\cO(|\phi|)}$ upper bound on the size of minimal models.

\begin{restatable}{proposition}{paramdmk}
	\label{proposition:param-dmk-model-size}
	Let $k\in\N$. If a sentence $\phi$ in $\ParamDMK{k}$ is satisfiable, then it has a finite model of size $2^{\cO(|\phi|)}$.
\end{restatable}

In other words, this means that, in Theorem~\ref{theorem:dmk-expo-size}, a more precise upper bound $2^{\cO(|\varphi|\cdot\log|\varphi|_\forall)}$ can be given, where $|\varphi|_\forall$ is the number of universal quantifiers in the formula $\varphi$.

We prove Proposition~\ref{proposition:param-dmk-model-size} in Appendix~\ref{appendix:skolem}.

\section{An application: Extending the uniform one-dimensional fragment}
\label{section:application}

In this section, we use Theorem~\ref{theorem:dmk-expo-size} to solve the satisfiability problem for the class $\FAUF$ which we define for the occasion: it is a strong extension of the fragment $\UF$ \cite{HK14}. 

Let us remark that in its original definition, $\UF$ neither allows equality nor constants. Its variant that admits equality was introduced in \cite{KK14}. In this version we present here, constants are allowed, but the equality symbol is not (as in $\DMK$).

Now, for a set of variables $\vs$, we say that a set $F$ of \emph{literals} (i.e. atoms or negations of atoms) 
is $\vs$-\emph{uniform} if the set of the variables of every literal in $F$ is precisely $\vs$.

\medskip
\noindent
{\bf Defition of UF$_{\mathbf{1}}$.}
The set of formulas of the \emph{uniform one-dimensional fragment}, \UF{},  is
defined as the smallest set such that: 

\begin{itemize}[nosep]
\item
Every atom with at most one variable is in $\UF{}$.
\item
It is closed under  Boolean combinations.
\item
Let $\ys$ be a set of variables, let $x$ be a variable not in $\ys$, and let $U$ be a set of formulas in \UF{} whose free variables are in $\{x\}\cup\ys$. Let $\vs\subseteq \{x\}\cup\ys$ and, finally, let $F$ be a $\vs$-uniform set of literals. Then, for every Boolean combination $\nu(x,\ys)$ of formulas in $U\cup F$, the two formulas $\exists x, \ys.~\nu(x,\ys)$ and $\exists \ys.~\nu(x,\ys)$ both belong to $\UF{}$.
\end{itemize}

Of course, since \UF{} is closed under negation, we allow ourselves to say that a formula such as $\forall x,\ys.~\nu(x,\ys)$ is a member of the class when $\neg\exists x,\ys.~\neg\nu(x,\ys)$ is. The same for the formula $\forall \ys.~\nu(x,\ys)$.

In \cite{HK14}, it is shown that \UF{} (without equality) has the doubly-exponentially-sized model property and hence its satisfiability is in \TwoNExpTime.
This is strengthened in \cite{KK14} to the exponentially-sized model property, and hence \NExpTime-completeness follows, even if free use of equality is allowed.

We define here an extension of \UF{}, which we call the $\forall$-\emph{uniform fragment} and denote by \FAUF{}. The idea is to keep the uniformity and one-dimensionality conditions in subformulas
starting with universal quantifiers, but not to require them in subformulas starting with existential quantifiers.

\smallskip
\noindent
{\bf Definition of $\mathbf{\forall}$-UF.}
	The class \FAUF{} of formulas, all in negation normal form, is defined as follows:

\begin{itemize}[nosep]
\item
Every literal with at most one variable is in $\FAUF{}$. 
\item
It is closed under positive Boolean combinations.
\item
Let $\ys$ be a set of variables, let $x$ be a variable not in $\ys$, and let $U$ be a set of formulas in \FAUF{} whose free variables are in $\{x\}\cup\ys$. Let $\vs\subseteq\{x\}\cup\ys$, and, finally, let $F$ be a $\vs$-uniform set of literals. Then, for every 
positive Boolean combination $\nu(x,\ys)$ of formulas in $U\cup F$, the formulas $\forall x, \ys.~\nu(x,\ys)$ and 
$\forall \ys.~\nu(x,\ys)$ both belong to $\FAUF{}$.
\item Let $\xs$ and $\ys$ be two disjoint sets of variables. Let $\nu(\xs,\ys)$ be a positive Boolean combination 
of literals containing at least one variable from $\ys$ and of formulas in \FAUF{} with free variables included in $\xs\cup\ys$. Then the formula $\exists \ys.~\nu(\xs,\ys)$ belongs to $\FAUF{}$. 
\end{itemize}

$\UF$ and $\FAUF$ differ in the last item, which allows us to use existential quantification quite freely, as
in the example $\varphi_{\textrm{lost\_proof}}$:
\begin{flalign*}
	\forall a,&s.~\big[\textrm{assertion}(a)\wedge\textrm{scientist}(s)\wedge \textrm{claims}(s,a)\big]\rightarrow\\
	&\exists p.~\textrm{proof\_of}(p,a)\wedge\textrm{found}(s,p)\\
	&\wedge\forall m.~\big[\textrm{margin}(m)\wedge\textrm{contains}(m,p)\big]\rightarrow\textrm{too\_small}(m).
\end{flalign*}

(Again, this specific use of implications is not problematic.)

Indeed, one can see that the subformula ``$\exists p\ldots$'' has two free variables, while the subformula ``$\forall m\ldots$'' has only one.

It is routine to verify that the negation normal form of any \UF{} formula is in \FAUF{}. It turns out that satisfiability of \FAUF{} (without equality) reduces to satisfiability of $\DMDK$ (and even to satisfiability of conjunctions of \Skolem-sentences).

\begin{restatable}{lemma}{expandingfauf}
	\label{lemma:expanding-fauf}
For every \FAUF{} sentence $\phi$, there is a conjunction $\psi$ of sentences in \Skolem, such that every model of $\phi$ can be expanded to a model of $\psi$; and reciprocally every model of $\psi$ is a model of $\phi$.
\end{restatable}

The basic idea of the translation is to replace in a bottom-up manner subformulas $\mu(x)=\forall \bar{y}.~\nu(x,\ys)$ of $\phi$ starting with a block
of universal quantifiers by unary atoms $P_\mu(x)$, and add
conjuncts $\forall x, \bar{y}.~ P_\mu(x) \rightarrow \nu(x,\ys)$, whose prenex form is in \Skolem. A formal proof of Lemma~\ref{lemma:expanding-fauf} can be found in Appendix~\ref{appendix:application}.
By Theorem~\ref{theorem:dmk-expo-size}, we get:
\begin{corollary}
The satisfiability problem for \FAUF{} is \NExpTime-complete. \FAUF{} has the exponential-sized model property.
\end{corollary} 
\section{Notes on related work} \label{s:discussion}
We present here technical remarks concerning previous works on the class \DMK{} and related fragments.

\medskip\noindent
{\bf Comments on the definition of \DMK{} from \cite{HuS99}.} \cite{HuS99} and its extended version in \cite{Hust99} thoroughly present a resolution
procedure for \DMK{}. There is however a minor issue concerning the definition of \DMK{} in those papers. Namely, according to this definition the
sentence
\begin{flalign*}
\forall y.~\exists z.~\forall x_1, x_2. \;
  & R(y,z) \wedge [R(x_1, x_2) \rightarrow \neg R(x_2, x_1)] \\
  \wedge &[\big(R(x_2,x_1) \wedge R(x_1, z)\big) \rightarrow R(x_2, z)]
\end{flalign*}
 is a member of \DMK{} (e.g. the \emph{terminal prefix} of the last atom is $\forall x_2$). The reader may check that it is satisfiable but has only infinite models.  On the other hand,
in \cite{HuS99} it is then claimed (without a proof) that sentences in \DMK{} can be converted to prenex form as in  (\ref{eq:prenex}) (with some extra leading
 existential quantifiers which we simulate by constants), and hence, by our result, the example sentence  should have a finite model.
One of the authors \footnote{U.~Hustadt, private communication, 2023.} of \cite{HuS99} confirmed that it was not their intention to have sentences like in our example (with alternation of quantifiers preceding the special universal quantification)  in \DMK{} and
hence their definition needed a fix.

\medskip\noindent
{\bf Complexity of the resolution procedure from \cite{HuS99}.} \cite{HuS99} and \cite{Hust99} do not study the complexity of the resolution procedure they present.
However, as both the number and the size of clause sets that need to be considered are at most doubly exponential with respect to the length of the formula, it seems
that with a careful implementation the procedure could work in doubly exponential time. Any better complexity does not seem to be  derivable.

\medskip\noindent
{\bf Prior decidability proofs for the \DMK-Skolem class.}
The \Skolem{} class was first studied, under the name ``Class 2.4'', in the book \cite{DG79}, among several other solvable Skolem classes.
In this book, on page $90$ it is shown that \Skolem{} has the finite model property, with a doubly exponential upper bound on the size of finite models.
Decidability in \TwoNExpTime{} follows. Another proof of the finite model property, via a probabilistic method, is given in \cite{Gol93}; neither the complexity nor the size of
models is  studied there.

It is worth to mention that earlier variants of this class were studied by Friedman~\cite{Friedman63}, and can also be found in Church's book~\cite{Church56}.

\medskip\noindent
{\bf Prior results on the G\"odel class.} There were a couple of proofs of the finite model property for the G\"odel class. See
\cite{BGG} for a comprehensive survey. Most of them, including G\"odel's original construction \cite{God33} and the probabilistic proof by Gurevich and Shelah \cite{GS83}, lead to a doubly exponential upper bound on the size of minimal finite models. Lewis \cite{Lew80} claims that from \cite{DG79} a bound  $2^{O(|\phi|)}$ follows. We however cannot see how to infer this claim from \cite{DG79}, as the only related statement we were able to find there is the one on page 94. It indeed speaks about an upper bound on the size of models, but, due to the factor
$d(F,2)$, taking into account the estimation on $d(F,n)$ from page 39, this upper bound seems to us to be doubly exponential.
We recall that the $2^{O(|\phi|)}$ upper bound follows from our work, namely from Proposition~\ref{proposition:param-dmk-model-size}.

\medskip\noindent
{\bf Prior results on  the Generalised Ackermann class.} The Generalised Ackermann class, GAF, was proposed by Voigt \cite{Voi19} as an extension of
the classical Ackermann fragment, AF, \cite{Ack28} (see page 67 of \cite{Voi19} for definitions). Voigt observed that GAF is contained in \DMK{} (Proposition 3.8.8, page 76)
and that every satisfiable GAF sentence has a finite model (Theorem 4.3.5, page 130). His finite model construction produces, however, models of
non-elementary size. Our results improve this upper bound to singly exponential and hence establish NExpTime-completeness of GAF satisfiability.
However, in contrast to Voigt we do not have equality. The exact complexity of GAF with equality remains open.

\medskip\noindent
{\bf Origins of Maslov's class \MK{}.}
In 1968, Maslov~\cite{Mas68} introduced the class \MK{} and proved the decidability of the corresponding validity problem by utilising the inverse method.
Originally his paper was written in Russian and its translation~\cite{Mas71} appeared later in 1971.

The syntactic restrictions of the class \MK{} at the first glance might appear artificial and rather unintuitive. However, the original paper sheds light on Maslov's motivation behind this definition.
When introducing the class \MK{}, he mentions two other classes, known to be decidable from earlier work.

The first one is the class in which every $\phi$-prefix has length at most $1$, e.g., Item IV in~\cite[p.~256]{Church56}.
This class naturally generalises the monadic class to higher arities by allowing repetitions of the same variable, e.g., $\forall x.~\exists y.~\forall z.~R(x,x) \wedge T(y,y,y) \vee \neg S(x) \wedge T(z,z,z).$

The second one is the class in which every non-empty $\phi$-prefix ends with a universal quantifier (remember, he considered the validity problem for the dual of $\DMK$).
In other words, the class \MK{} with a restriction: the tuple of special variables shall be empty.
Maslov claims that this class was probably not known for a wider audience, however, he points to Item IX in~\cite[p.~257]{Church56}, corresponding to a significant subfragment of it:
the class of formulas in the form $\Qfr_1 x_1 \dots \Qfr_\bK x_\bK.~\forall y_1 \dots \forall y_\bM.~\psi$ with each non-empty $\phi$-prefix ending with a variable $y_i$, for some $i$.

Additionally, he drew inspiration from the G\"odel class and was aware of the solvable Skolem classes in~\cite{Friedman63} and in~\cite{Church56}, notably leading him to consider tuples of special variables in his fragment.

With all of this in mind, we believe that the class \MK{} can be considered as a reasonable generalisation of the other known decidable fragments, combining the ideas present there in a clever way.

\medskip\noindent
Sergey Yurevich Maslov was a Russian logician employed in the Steklov Mathematical Institute in St.~Petersburg.
He mainly worked in the field of automated reasoning and its applications, where he contributed many interesting ideas.
However, his interests were not limited to logic and included cognitive science, mathematical biology, economics and philosophy.
He was also a civil right activist during the Soviet regime.
Maslov died unexpectedly in 1982, due to a car accident, at the age of 43. (Biographical notes in \cite{Lifschitz89}.)

\begin{acks}
The first and second authors were supported by Polish National Science Center grant No. 2021/41/B/ST6/00996. The third author was supported by the ERC grant INFSYS, agreement no. 950398.
The authors thank Warren Goldfarb, Ullrich Hustadt and Harry R. Lewis for some comments on their work, and the anonymous reviewers for their helpful suggestions.
\end{acks}

\bibliographystyle{ACM-Reference-Format}
\bibliography{mybib} 

\appendix

\section{Interpreting constants} 
\label{appendix-constants}

In our definitions we  assume that constants are interpreted by themselves; in particular different constant symbols are interpreted by different elements.
This does not affect the generality of our results, as we have a size preserving reduction based on the following: 
\begin{proposition}\label{fact:reduction-distinct-const-sat}
	Let $\phi$ be a satisfiable $\FO$-sentence, and let $c_1,\dots,c_k$ be the list of constant symbols occurring in $\phi$.
	Then there exists $\psi$, with $\sigma(\psi) \subseteq \sigma(\phi)$, differing only in constant symbols, such that:
	\begin{enumerate}
	\item[(i)] $|\psi|=|\phi|$,
	\item[(ii)] $\psi$ is satisfiable assuming distinct interpretation of constants,
	\item[(iii)] any model $\str{B} \models \psi$ can be expanded to a model $\str{B}^+ \models \phi$ over the same domain.
	\end{enumerate}
\end{proposition}
\begin{proof}
	Suppose that $\str{A} \models \phi$. We define $\iota$ as the mapping $i \mapsto \min\{ j : \str{A} \models c_i = c_j \}$, and $\psi$ as the formula 
	obtained from $\phi$ by replacing each occurrence of $c_i$ by $c_{\iota(i)}$. Clearly $|\psi|=|\phi|$. Then $\str{A} \models \psi$ and all the constant symbols in $\psi$ have distinct interpretations.
	
	Finally, suppose that we have $\str{B}$ over $\sigma(\psi)$ such that $\str{B} \models \psi$.
	Construct $\str{B}^+$ from $\str{B}$ by adding the missing interpretations of constants: if $c_i \not\in \sigma(\psi)$ then set
	$c^{\str{B}^+}_i:=c^\str{B}_{\iota(i)}$. 
	 Then $\str{B}^+ \models \phi$.
\end{proof}

Hence, the (unrestricted) satisfiability problem of $\DMK$ can be reduced in nondeterministic polynomial time to the satisfiability problem of the same fragment, under the
assumption that different constants have different interpretations. 
\section{An explicit construction of paradoxical colourful tournaments}
\label{appendix:graphs}

The existence of paradoxical colourful tournaments of exponential size is already established by Lemma~\ref{lemma:colourful-tournaments}.
However, the proof given in Section~\ref{section:graphs} relies on a probabilistic method, hence it is non-constructive.
We complement this result by providing a deterministic construction.

Our proof is based on Paley graphs. Their exact definition is not needed to understand our proof, however, for self-completeness we briefly recall it in the following paragraph.

Let $p$ be a prime power such that $p \equiv_4 3$. Denote by $\mathbb{F}_p$ the finite field of order $p$.
Then the \emph{Paley graph of size $p$} is the graph having the elements of $\mathbb{F}_p$ as its set of vertices,
and such that, for any distinct $a,b \in \mathbb{F}_p$, there is an arc $a \arc b$ iff there exists $c \in \mathbb{F}_p$ such that $a-b=c^2$.
We remark that, because of the condition $p \equiv_4 3$, such a graph is a tournament (see \cite{graham_spencer_1971}), which is why we call it the \emph{Paley tournament} in the rest of the section.

We say that a tournament has the \emph{$k$-extension property} if for any disjoint subsets of vertices $A$ and $C$, such that $|A|+|C|\le k$,
there exists a vertex $b \not\in A \cup C$ such that $c \arc b$ for each $c \in C$ and $b \arc a$ for each $a \in A$.
The main result of \cite{graham_spencer_1971} states that:
\begin{theorem}[\cite{graham_spencer_1971}]
	Any Paley tournament of size at least $k^2\cdot 2^{2k}$ admits the $k$-extension property.
\end{theorem}

To be more precise, the statement in~\cite{graham_spencer_1971} only tells about $k$-paradoxicality, but the $k$-extension property can be inferred from the involved proof.

\medskip
\noindent
{\bf Explicit construction.}
The high-level idea behind our construction is to represent both vertex and arc labellings directly inside Paley tournaments: any information such as a vertex colour, or an arc colour, can be encoded into a series of bits, and the sets $A$ and $C$ from the definition of the $k$-extension property can be taken as the sets of bits that shall be put to $0$ and $1$ respectively.
Hence, let us define a suitable gadget for representing a single integer:

Suppose that $\bar{a}=\tuple{a_1,\dots,a_\ell}$ is a tuple of pairwise distinct vertices in a tournament, $b$ is a vertex not in $\bar{a}$, and $n$ is a number from $\{0,\dots,2^{\ell}-1\}$, where $\ell$ is some natural number. The number $n$ can be written in binary as $\Sigma_{1\leq i \leq \ell}{n_i\cdot 2^{i-1}}$, where each $n_i$ is a bit $\in\{0,1\}$.
Then we say that the pair $(\bar{a},b)$ \emph{represents} the number $n$ if, for each $i\in[\ell]$, the following equivalence holds: $a_i \arc b$ iff $n_i=1$ (and therefore $b\arc a_i$ iff $n_i=0$). We denote this condition by $\mask(\bar{a},b)=n$.

Now, let us come back to our colourful tournaments. For simplicity, we assume our sets of vertex colours and of arc colours to be respectively $\cR=\{0,1,\ldots, 2^m{-}1\}$ and $\cQ=\{0,1,\ldots, 2^{2^t-1}{-}1\}$, where $t$ and $m$ are natural numbers. We set $k=t + m + 2^{2^t-1} \cdot 2^t$.

As a foundation of our construction of an $(\cR,\cQ)$-paradoxical colourful tournament, we take a Paley tournament $\cP = (V_p,E_p)$ of size $\cO(k^2\cdot 2^{2k})$ with the $k$-extension property.

We partition the vertex set $V_p$ into three distinct subsets as follows: $\CBit$ of size $t$, $\CRole$ of size $m$, and $\cU$, the rest of them, of size $p{-}t{-}m$.
We can rewrite the elements of $\CBit$ as $\bit_1,\ldots, \bit_t$, and the elements of $\CRole$ as $\mu_1,\ldots, \mu_m$.
Furthermore, we partition the set $\cU$ into classes indexed by $n \in \{0,\dots,2^t{-}1\}$ as follows: $$\cC_{n} = \{ x \in \cU : \mask(\CBit,x)=n \}.$$
Then, we rename these classes $\cC_n$'s with indices $i \in [2^t]$, so that $|\cC_1| \le |\cC_2| \le \dots \le |\cC_{2^t}|$. 
We define $n^*$ to be the original index of the smallest class $\cC_1$, i.e. $n^* = \mask(\CBit,x)$ for any $x \in \cC_1$.

Now, we consider any subset $\cS$ of $\cC_1\times \cC_2\times\cdots\times \cC_{2^t}$
such that each $x \in V_p$ occurs in at most one tuple in $\cS$ and there is no $\cS' \supset \cS$ satisfying the same property.
Since $\cC_1$ is the smallest class, it can be easily seen that $|\cS|=|\cC_1|$, and that for each $a_1 \in \cC_1$, there is a unique tuple $\tuple{a_2,\dots,a_{2^t}}$ such that $\tuple{a_1,a_2,\dots,a_{2^t}} \in \cS$. We denote this $(2^t{-}1)$-tuple by $s(a_1)$.

Finally, we can define a new tournament $\cT = (V,E)$ using this decomposition of $\cP$ as follows: the set $V$ of vertices is $\cC_1$, and $E$ is $E_p\cap \big(\cC_1{\times}\cC_1\big)$, the induced set of arcs from $\cP$.
We turn this tournament into an $(\cR,\cQ)$-colourful one by defining $\mu(b) = \mask(\CRole,b)$ for $b\in V$ (recall indeed that $\mask(\CRole,b)$ is a number in $\{0,\ldots 2^{m}{-}1\}$), and, if $b \arc b'$ is an arc, then $\lambda(b \rightarrow b') = \mask(s(b'),b)$ (same remark: $\mask(s(b'),b)$ is a number in $\{0,\ldots, 2^{2^t{-}1}{-}1\}$).

\begin{claim}
  The so-defined triple $(\cT,\mu,\lambda)$ is indeed an $(\cR,\cQ)$-paradoxical colourful tournament.
\end{claim}

\begin{proof}
  For simplicity, we will write $\ell$ for $2^{2^t-1}$.
  
  To prove the $(\cR,\cQ)$-paradoxicality of $\cT$, we must in particular show that $V=\cC_1$ admits at least $\ell=|\cQ|$ vertices. This will be done in the end of this proof.
   
  For the moment, we assume that $\cC_1$ indeed admits at least $\ell$ elements, and we fix a vertex colour $r \in \cR$, a tuple $\bar{a}=\tuple{a_1,\dots,a_\ell}$ of pairwise distinct vertices of $\cT$, and a tuple $\bar{q}=\tuple{q_1,\dots,q_\ell}$ of arc colours.
  We need to prove the existence of a vertex $b$ in $V$ that colourfully dominates $\bar{a}$ via $r$ and $\bar{q}$, that is (let us recall the definition) such that $\mu(b)=r$ and $b\arc a_i$, $\lambda(b\arc a_i)= q_i$ for each $i\in[\ell]$.
  For this, we will use the $k$-extension property of $\cP$.
  
  We define $A$ as the set $\bar{a}\cup X_0\cup R_0\cup\bigcup_{1\leq i\leq \ell} Q_{0,i}\subseteq V_p$ and $C$ as the set $X_1\cup R_1\cup \bigcup_{1\leq i\leq \ell}{Q_{1,i}}\subseteq V_p$, where:
  \begin{itemize}[nosep]
    \item for $d\in \{0,1\}$, $X_d$ is the set of elements $\bit_j\in\CBit$ such that $n^*_{j}$, the $j$'th bit of $n^*$ in its binary representation, is $d$ (remember that $n^*$ is a member of $\{0,\ldots, 2^t{-}1\}$);
    \item for $d\in \{0,1\}$, $R_d$ is the set of elements $\mu_j\in\CRole$ such that $r_j$, the $j$'th bit of $r$ in its binary representation, is $d$ (remember that $r$ is a member of $\cR=\{0,\ldots, 2^{m}{-}1\}$);
    \item for $1\leq i\leq \ell$, let us write $\langle a_{i,2},\ldots, a_{i,2^t}\rangle\in \cC_2\times \cdots\times\cC_{2^t}$ for $s(a_i)$, then $Q_{d,i}$, for $d\in\{0,1\}$, is the set of $a_{i,j}$'s such that $q_{i,j}$, the $j$'th bit of $q_i$ in its binary representation, is $d$ (remember that $q_i$ is a member of $\cQ=\{0,\ldots, 2^{2^t-1}-1\}$).
  \end{itemize}

  It is easy to see that $A$ and $C$ are disjoint and are both contained in the set $\bar{a} \cup \CBit \cup \CRole \cup \bigcup_{1\leq i\leq \ell}{s(a_i)}$, so $|A|+|C| \le \ell + t + m + \ell \cdot (2^t{-}1) = k$.
  Thus, we can apply the $k$-extension property of $\cP$ to find a vertex $b\in V_{p}$ such that $b \arc a$ for each $a\in A$ and $c \arc b$ for each $c\in C$.
  Now, we verify that $b$ is a good candidate, i.e. that it is a member of $\cT$, in which it colourfully dominates $\bar{a}$ via $r$ and $\bar{q}$.
  
  First, $b\in \cC_1$, i.e. $b$ is indeed a vertex of $\cT$.
  This is ensured by our choices of $X_0 \subseteq A$ and of $X_1 \subseteq C$: they guarantee that $\mask(\CBit,b) = n^*$, and hence $b \in \cC_1$ by definition.
  
  Second, the vertex colour of $b$ is indeed $r$. By definition, we have $\mu(b) =\mask(\CRole,b)$, and the latter is $r$, by our choices of $R_0 \subseteq A$ and of $R_1 \subseteq C$.
  
  Third, $b$ indeed dominates $\bar{a}$, that is $b \arc a_i$ for each $1 \leq i \leq \ell$. This is ensured by the fact that $\bar{a}$ is included in the set $A$.

  Finally, $\lambda(b \arc a_i)$ must be $q_i$ for each $i\in[\ell]$. By definition, it is $\mask(s(a_i),b)$, which is $q_i$ by our choices of $Q_{0,i} \subseteq A$ and of $Q_{1,i} \subseteq C$.

  Now, as discussed above, it remains to prove that $V$ admits at least $\ell$ elements. Suppose it is not the case. Then $V_p \setminus V$ admits at least $t{+}2^t{-}1$ elements.
  Let $\tuple{y_1,\dots,y_{2^t{-}1}}$ be a tuple of elements of $V_p \setminus (V \cup \CBit)$.
  For each $n \in \{ 0, \dots, \ell{-}1 \}$ and $d \in \{0,1\}$, we define $Y_{d}^{(n)}$ to be the set of $y_j$'s such that $n_{j}$, the $j$'th bit of $n$ in its binary representation, is $d$.
  Finally, we define $A^{(n)}$ to be $X_0 \cup Y_{0}^{(n)}$ and $C^{(n)}$ to be $X_1 \cup Y_{1}^{(n)}$, where $X_d$'s are as before.

  Since $A^{(n)}$ and $C^{(n)}$ are disjoint and $|A^{(n)}|+|C^{(n)}| \le k$, we can apply the $k$-extension property of $\cP$ to find a vertex $b_n \in V_p$ such that $ b_n \arc a$ for each $a \in A^{(n)}$ and $ c \arc b_n $ for each $c \in C^{(n)}$.
  As we shown earlier, because of $X_0 \subseteq A^{(n)}$ and $X_1 \subseteq C^{(n)}$, we have that $b_n \in \cC_1$, i.e. $c_n \in V$.
  Moreover, for $n \neq n'$, the vertices $b_n$ and $b_{n'}$ are necessarily distinct.
  Therefore, the size of $\cT$ is at least $\ell$, which completes the proof.
\end{proof}

It remains to estimate the size of the tournament $\cT = (V,E)$.
Let $n$ denote the size of $\cP = (V_p,E_p)$, which is $\cO(k^2\cdot 2^{2k})$. Since $k = t + m + 2^{2^t-1} \cdot 2^t$,
we come by a simple calculation to $|V_p|= 2^{\cO(|\cQ|\cdot \log|\cQ|)}\times(|\cR|\cdot\log|\cR|)^2$. Since $V$ is a subset of $V_p$, we can conclude.

This way, we obtained an alternative fully constructive proof of Lemma~\ref{lemma:colourful-tournaments}.
Although, the resulting upper bound is slightly worse, compared to the previous $2^{\cO(|\cQ|\cdot\log|\cQ|)}\times|\cR|\cdot\log|\cR|$ bound obtained via the probabilistic method,
we still consider it as quite good for our purpose.

\section{Proofs from Section~\ref{section:satgames}}
\label{appendix:satgames}

In this appendix, we provide the proofs of Lemma~\ref{l:satandgames} and of Item~\ref{it:corresp_small_v} of Claim~\ref{c:correspondance1}.

\satandgames*

\begin{proof} 
Let $\str{A}_0$ be a model of $\varphi$. We first modify $\str{A}_0$ to obtain an infinite structure $\str{A}$ with the unnamed domain $A = A_0^\Cons {\times} \N$.
(Notice that $A$ contains unnamed copies of constants.)
For every $a \in A^\Cons$, we define its \emph{pattern element} $a_*$ to be $a$ if $a$ is in $\Cons$, and $a_0$ if $a = (a_0,i)$ for some $i \in \N$.
  Finally, we define relations in $\str{A}$ in the following way: for each relational symbol $R$, and every $a_1, \ldots, a_k \in A^\Cons$, where $k=\arity(R)$, we have $\str{A} \models R(a_1,\ldots,a_k)$ iff $\str{A}_0 \models R\big((a_1)_*, \ldots, (a_k)_*\big)$.

As $\phi$ does not use equality it follows that $\str{A} \models \phi$. Moreover, Eloisa can win $\VER(\phi, \str{A})$
using a \emph{proper} winning strategy, that is, when extending the assignment to a new variable, she can always choose a \emph{fresh} unnamed element (i.e. an element not assigned to any variable in earlier rounds).

Indeed, in each Round $t{+}1$ belonging to Eloisa (i.e. $\Qfr_{t+1}=\exists$), when choosing her move after a position $f_t\colon {\Vars}_t \rightarrow A^\Cons$ is reached in $\VER(\phi,\str{A})$, she can consider the corresponding position $f'_t\colon {\Vars}_t \rightarrow A_0^\Cons$ in
$\VER(\phi, \str{A}_0)$, where $f'_t(v)$ is defined to be $\big(f_t(v)\big)_*$ for every $v \in {\Vars}_t$. Then she looks at her winning extension $f'_{t+1}$ of $f'_t$,
  and extends $f_t$ to $f_{t+1}$ by taking as its value on $y_{t+1}$ a pair $(f'_{t+1}(y_{t+1}), i)$, with some $i \in \N$ such that $(a_0,i)\not\in f_t[\Vars_t]$ for any $a_0\in A_0^\Cons$. 

Let $\BBB$ be the set of outer-types of grade at most $\bK{+}\bM$ and realised in $\str{A}$. It is readily verified
that $\BBB$ is consistent and closed.

We demonstrate now that Eloisa has a winning strategy $\omegas$ in $\SAT(\phi, \BBB)$. 
To find such a strategy, we fix her winning proper strategy $\omegav$ in $\VER(\phi, \str{A})$, and naturally ``simulate'' plays of $\SAT(\phi, \BBB)$ as plays of $\VER(\phi, \str{A})$ consistent with
$\omegav$.
By simulation we mean here that in parallel to the play of $\SAT(\phi,\BBB)$, we will be constructing a play of $\VER(\phi,\str{A})$, mimicking players choices in $\SAT(\phi,\BBB)$.

  Let us consider Abelard's opening move $(\str{L}_0, f^\textrm{S}_0)$, where $\str{L}_0\in\BBB$ is a $k$-outer-type for some $k\in\N$, and $f^\textrm{S}_0\colon\Vars_0 \rrightarrow [k]^\Cons$ is an assignment. We select a $k$-tuple $\bar{a}=\langle a_1, \ldots, a_k \rangle$ of distinct unnamed elements of $\str{A}$ realising $\str{L}_0$, i.e. $\outertype{\str{A}}{a_1, \ldots, a_k}$ is $\str{L}_0$.
  We define the assignment $f_0^\textrm{V}\colon\Vars_0 \rightarrow A^\Cons$ mimicking $f^\textrm{S}_0$: if $f^\textrm{S}_0(v)=c$ for some constant $c$, then $f^\textrm{V}_0(v)=c$; 
  if $f^\textrm{S}_0(v)=i$ for some $i\in[k]$, then $f^\textrm{V}_0(v)=a_i$. This way we obtained a position $f^\textrm{V}_0$ in $\VER(\phi,\str{A})$, simulating Abelard's opening move.

  Suppose that after Round $t$ a position $(\str{L}_t,f_t^\textrm{S})$ is reached in $\SAT(\phi,\BBB)$, with $L_t=[k]$.
  Simulating $\VER(\phi,\str{A})$ in Rounds $0,\dots,t$ as we are describing, we obtain also a position $f_t^\textrm{V}\colon{\Vars}_t\rightarrow A^\Cons$.

  If Round $t{+}1$ belongs to Abelard (i.e. $\Qfr_{t+1}=\forall$) and he decides to add a new element $k{+}1$ as $f^\textrm{S}_{t+1}(y_{t+1})$, we simulate his move in $\VER(\phi,\str{A})$ by taking as $f^\textrm{V}_{t+1}(y_{t+1})$ a fresh element of the same $1$-type in $\str{A}$. Naturally, if he chooses to reuse an element from $f^\textrm{S}_t[\Vars_t]\cup\Cons$, we reuse the corresponding element from $f^\textrm{V}_t[\Vars_t]\cup\Cons$ as well.

In the opposite case, when Eloisa is to make a move (i.e. $\Qfr_{t+1}=\exists$), her strategy $\omegas$ in $\SAT(\phi,\BBB)$ works as follows.

Let $f_{t+1}^\textrm{V}$ be $\omegav(f_t^\textrm{V})$, i.e. the position suggested by $\omegav$ in $\VER(\phi,\str{A})$, and let $a$ be $f^\textrm{V}_{t+1}(y_{t+1})$.
In $\SAT(\phi,\BBB)$, Eloisa extends the structure $\str{L}_t$ by setting the type of the fresh element $k{+}1$ to the same $1$-type as $a$, $\tp^{\str{A}}(a)$; then she sets the required hull-types of
the tuples containing $k{+}1$ by copying them from  the corresponding tuples of $\str{A}$ containing $a$.

If $(\str{L}_M, f^\textrm{S}_M)$ is the finally reached position in $\SAT(\phi, \BBB)$, and $f^\textrm{V}_M$ is the position reached in the simulation of $\VER(\phi, \str{A})$, then, for every atom
  $\gamma(\vs)$ of $\psi$, we have that $\str{L}_M, f^\textrm{S}_M \models \gamma(\vs)$ iff $\str{A}, f^\textrm{V}_M \models \gamma(\vs)$.
Hence, the described strategy $\omegas$ is indeed winning for Eloisa.
\end{proof}

Now, we prove Item~\ref{it:corresp_small_v} of Claim~\ref{c:correspondance1}: let $\rho_t'=(\str{L}'_t, f'_t)$ be a position of order $t$ that can be reached when following Eloisa's strategy $\omegas'$, and let $\rho_t=(\str{L}_t, f_t)$ be $\Gamma(\rho_t')$. Then, for every atom $\gamma(\vs)$ of $\psi$ such that $\vs\subseteq{\Vars}_t$, we have that $\str{L}'_t,f_t\models\gamma(\vs)$ iff $\str{L}_t,f_t\models \gamma(\vs)$.

\begin{proof}
	We proceed by induction on $t$.
	
	If $t=0$, then the domain of $f_t$ ($=f'_t$, by Item~\ref{it:corresp_small_i}) is ${\Vars}_0$, so the atoms $\gamma(\bar{v})$ we are interested in either contain all the $x_i$'s or at most one of them.
	
	If $|L_t| \le 1$ then $\str{L}_t=\str{L_t}'$ by our definition of $\Gamma$, hence both positions are the same, and the claim trivially holds.
	
	Suppose now that $|L_t|=k>1$.
	
    Let $\gamma(\xs)$ be any atom of the first kind (it contains all the $x_i$'s), the values of $\gamma(f_t(\xs))$ in both structures coincide since $\hulltype{\str{L}_t}{1, \ldots, k}=\hulltype{\str{L}'_t}{1, \ldots, k}$ by Item~\ref{it:corresp_small_iv} (recall that $f_t$ is onto $L_t$).

	Let now $\gamma(\vs)$ be any atom of the second kind (it contains at most one of the $x_i$'s). If $\vs=\emptyset$, then the equivalence trivially follows from $\str{L}'_t\restr\emptyset=\str{L}_t\restr\emptyset$ (Item~\ref{it:corresp_small_i}). Assume now that $\vs=\{x_i\}$, and let us consider two subcases. If $f_t(x_i) \in \Cons$, then again the equivalence follows from~\ref{it:corresp_small_i}; otherwise the values of $\gamma(f_t(\bar{v}))$ in both structures
	are determined by the $1$-types of $f_t(x_i)$ in these structures. By Item~\ref{it:corresp_small_ii}, these $1$-types are in relation $\sim_{f^*}$ for some $f^{*}$,
	and
	the first condition in the definition of  $\sim_{f^{*}}$ gives as that the values of $\gamma(f_t(\bar{v}))$ in both structures are identical.
	
	Assume now that the claim holds for positions of order $t$, and consider a position $\rho_{t+1}'=(\str{L}'_{t+1},f'_{t+1})$ of order $t{+}1$ reachable in $\SAT(\varphi,\BBB')$ by following $\omegas'$, let $\rho_{t+1}=(\str{L}_{t+1},f_{t+1})$ be its image $\Gamma(\rho'_{t+1})$.
	
	By the inductive
	assumption, $\rho_{t+1}$ and $\rho'_{t+1}$ agree on the atoms containing at most the variables from ${\Vars}_{t}$. We need to consider now
	an atom $\gamma(\bar{v})$ such that $y_{t+1}\in\vs\subseteq {\Vars}_{t+1}$.
	
	If Round $t{+}1$ is universal then, since $y_{t+1}$ is universally quantified, but not special, $\gamma$ may contain at most one variable, and we reason as in the basis of induction, using Item~\ref{it:corresp_small_ii}.

	If Round $t{+}1$ is existential, assume that $a$ is the fresh element added in this round. There are two subcases. If  $a$ is the only unnamed element in $f_t(\bar{v})$ then the truth-values of $\gamma(f_t(\bar{v}))$ agree in both positions by~\ref{it:corresp_small_iii}, using Condition~\ref{simfitem2} from the definition of  $\sim_{f}$. If $f_t(\bar{v})$ contains some other unnamed element then
    $\gamma(f_t(\bar{v}))$ agree in both positions by the fact that, by our construction (Item~\ref{it:corresp_small_iv}), the hulls of the outer-types containing $a$
	are identical in both structures.
\end{proof}

\section{Missing details from Section~\ref{section:lower-bounds}}
\label{appendix:constant-free-lowerbound}

In this appendix, we first give the proofs of the claims~\ref{claim:phi-n-satisfiable} and~\ref{claim:phi-n-lowerbound}, proving the desired properties of the formula $\varphi_n$ defined in Section~\ref{section:lower-bounds}.
Then, we explain the modifications that can be done so that this formula does not mention any constant symbols (as in the statement of Proposition~\ref{proposition:phi-n-constantfree}).

\subsection{Proofs of the claims~\ref{claim:phi-n-satisfiable} and~\ref{claim:phi-n-lowerbound}}

\phinsatisfiable*

\begin{proof}
	
	The fact that $\varphi_n$ is in $\Skolem$ comes from its very definition: it belongs to the fragment $\forall^\ast\exists$, and the reader can check that the set of variables of any of its atoms is exactly the set $\bar{x}\cup \{y\}\cup\bar{r}$, unless it contains the existentially quantified variable $w$.

    The formula is of size $\cO(n)$, as all of the formulas $\mu_{n,\textrm{perm}}$, $\mu_{n,\textrm{witness}}$, etc. are.
	
	Finally, $\varphi_n$ is satisfiable. Indeed, the prototypical model $\fA$ of $\varphi_n$ is defined as follows:
	
	\begin{itemize}[nosep]
		\item each constant $c_i$ is interpreted as itself, the same for $q_0$ and $q_1$; 

		\item the set $A$ of unnamed elements is the set $\Perms_n$ of permutations of $[n]$;

		\item $P^\fA$ is the set of tuples of the shape $$\langle c_{\pi(1)},\ldots,\break c_{\pi(n)}\mid q_0\mid q_0, \ldots, q_0\rangle,$$ with $\pi$ ranging over $\Perms_n$;
		
		\item $W^\fA$ is the set of tuples of the shape $$\langle c_{\pi(1)},\ldots,\break c_{\pi(n)}\mid \pi\mid q_0,\ldots, q_0\rangle,$$ with $\pi$ ranging over $\Perms_n$;
		
		\item $C_\rightarrow^\fA$ is the set of tuples of the shape $$\langle c_{\pi'(1)},\ldots,\break c_{\pi'(n)}\mid \pi\mid q_1, \ldots, q_1, q_0,\ldots,q_0 \rangle,$$ where $\pi$ ranges over $\Perms_n$, $\pi'$ ranges over the permutations obtained as $\gamma^k\circ\pi$, $0<k<n$, and the number of occurrences of $q_1$ in the tuple is $k$;
		
		\item $S^\fA$ is the set of tuples of the shape $$\langle c_{\pi'(1)},\ldots,\break c_{\pi'(n)}\mid \pi\mid q_1, \ldots, q_1, q_0,\ldots,q_0 \rangle,$$ where $\pi$ ranges over $\Perms_n$, $\pi'$ ranges over the permutations obtained as $\rho\circ\gamma^k\circ\pi$, $\rho\in\Perms_{n}$ assigning $n$ to itself, $0<k<n$, and the number of occurrences of $q_1$ in the tuple is $k$;
		
		\item $C_\leftarrow^\fA$ is the set of tuples of the shape $$\langle c_{\pi'(1)},\ldots,\break c_{\pi'(n)}\mid \pi\mid q_1, \ldots, q_1, q_0,\ldots,q_0 \rangle,$$ where $\pi$ ranges over $\Perms_n$, $\pi'$ ranges over the permutations obtained as $\gamma^{-j}\circ\rho\circ\gamma^k\circ\pi$, $\rho\in\Perms_{n}$ assigning $n$ to itself, $0\le j<k<n$, and the number of occurrences of $q_1$ in the tuple is $k{-}j$;
		
		\item $Z^\fA$ is the set of tuples of the shape $$\langle c_{\pi'(1)},\ldots, c_{\pi'(n)}\mid \pi\mid q_1, \ldots, q_1, q_0,\ldots,q_0 \rangle,$$ where $\pi$ and $\pi'$ are defined with the same conditions,  and the number of occurrences of $q_1$ in the tuple is at most $k{-}j$.
	\end{itemize}
	
\end{proof}

\phinlowerbound*

\begin{proof}

  Let us consider any model $\fA$ of $\varphi_n$.
  The cardinality of $\fA$ being at least $n!=2^{\Omega(n\log n)}$ comes from the construction:
	\begin{itemize}
		\item for any permutation $\pi$ of $[n]$, we have $$\langle c_{\pi(1)}, \ldots,\break c_{\pi(n)}\mid q_0\mid q_0, \ldots, q_0\rangle \in P^{\fA};$$
		
		\item therefore, there exists some element $w_\pi\in A^\Cons$ such that $$\langle c_{\pi(1)}, \ldots, c_{\pi(n)}\mid w_\pi\mid q_0, \ldots, q_0\rangle \in W^{\fA};$$
		
		\item then, for any $\pi'$ distinct from $\pi$, we get $$\langle c_{\pi'(1)},\ldots,\break c_{\pi'(n)} \mid w_\pi\mid q_1, \ldots, q_1, q_0,\ldots, q_0 \rangle\in C_\leftarrow^\fA,$$ for some number of occurrences of the constant $q_1$;
		
		\item finally, it means that the tuple $\langle c_{\pi'(1)},\ldots, c_{\pi'(n)} \mid w_\pi\mid q_0,\ldots, q_0\rangle$ is in $Z^{\fA}$, and therefore $$\langle c_{\pi'(1)}, \ldots, c_{\pi'(n)}\mid w_\pi\mid q_0, \ldots, q_0\rangle \notin W^{\fA}.$$
	\end{itemize}
	
    Therefore, we can deduce that for every distinct permutations $\pi$ and $\pi'$ the elements $w_{\pi}$ and $w_{\pi'}$ are also distinct. 

    By the way, we can also justify that the interpretations of $q_0^\fA$ and $q_1^\fA$ also have to be distinct. Indeed, if it is not the case, then the unary counter in the arguments is always the same tuple $\tuple{q_0^\fA,\dots,q_0^\fA}$. In particular, we get that $\langle c_{\pi(1)},\ldots, c_{\pi(n)} \mid w_\pi \mid q_1, \ldots, q_1, q_0,\ldots, q_0 \rangle \in C_\leftarrow^\fA$, and hence $\langle c_{\pi(1)}, \ldots, c_{\pi(n)}\mid w_\pi\mid q_0, \ldots, q_0\rangle \notin W^{\fA}$, which goes in contradiction with the second point.

    A similar argument shows that the interpretations of $c_i^\fA$'s are necessarily pairwise distinct too. Suppose that two indices $j<k$ are such that $c_j$ and $c_k$ have the same interpretation in $\fA$. As written above, the atom $W(c_1^\fA, \ldots, c_n^\fA\mid w_{\textrm{id}}\mid q_0, \ldots, q_0)$ holds in $\fA$, but for any $\pi$ distinct from the identity, the atom $W(c_{\pi(1)}^\fA, \ldots, c_{\pi(n)}^\fA\mid w_{\textrm{id}}\mid q_0, \ldots, q_0)$ does \emph{not} hold. We immediately see the contradiction when $\pi$ is the permutation switching $j$ and $k$, since in that case, the tuple of the $c_{\pi(i)}^\fA$'s matches exactly with the tuple of the $c_i^\fA$'s. Therefore, all the $c_i^\fA$'s must be distinct.
\end{proof}

\subsection{Tight examples without constant symbols}

In the rest of this appendix, we show how to modify the formula $\varphi_n$ so that it does not mention constant symbols.

Let us recall briefly how the original formula $\varphi_n$ was working: we simulated the set $\Perms_n$ of permutations of $[n]$ via the distinct constants $c_1,\ldots, c_n$, and called a distinct witness for each permutation. The formula $\varphi_n$ was obtained as $\forall x_1,\ldots, x_n,y, r_1, \ldots, r_{n-2}.~\exists w.~\psi_n$, and in $\psi_n$ were the following relational symbols, all of arity $n{+}1{+}n$:
\begin{itemize}[nosep]
\item[--] $P$, that simulates permutations of the set $[n]$;
	\item[--] $W$, that provides a witness for every permutation of $[n]$;
	\item[--] $C_\rightarrow$, $S$, and $C_\leftarrow$, that generate distinct permutations from the one considered, with the help of a counter;
	\item[--] $Z$, that finally sets said counter back to zero and forbids any other permutation for the same witness.
\end{itemize}
Finally, the counter was written in unary, with the use of the constants $q_0$ and $q_1$.

The main difference, in comparison to the original formula, is that, we reintroduce $q_0$ and $q_1$ as new variables, quantified universally: the sequence of quantifiers of our new formula is $\forall x_1,\ldots, x_n,y, q_0, q_1, r_1, \ldots, r_{n-2}.~\exists w$, and all the variables, except $w$ will be special.

We invite the reader to verify that every atom in the different formulas introduced below indeed have variables satisfying the conditions of Maslov's class $\DMK$.

Alongside turning $q_0$ and $q_1$ into variables, we introduce a new relational symbol $U$, of arity $1$, as it will distinguish the elements simulating constants $q_0$ and $q_1$:
the role of $q_0$ will be fulfilled by \emph{any} element not satisfying $U$, and the role of $q_1$ by \emph{any} element satisfying $U$.

On the other hand, the constants $c_i$'s will not be represented explicitly in the variables.
Instead, to obtain elements fulfilling the roles of $c_i$'s, we make use of yet another new relational symbol $F$, of arity $2n{+}3$.
For readability, its arguments will be sequenced in four blocks of respective lengths $n$, $1$, $2$, and $n$.

What we aim here is the satisfaction of an atom $F(x_1,\ldots, x_n\mid q_0\mid q_0, q_1\mid q_0,\ldots, q_0)$, where the $x_i$'s are distinct elements, and $q_0$ (resp. $q_1$) does not (resp. does) satisfy $U$. Such an atom could then be used as a starting point for defining the different permutations (recall that in the original formula, the first atom introduced was $P(c_1,\ldots, c_n\mid q_0\mid q_0,\ldots, q_0)$).

This is done in four steps. First, we ensure the existence of an element, call it $q_0^{\fA}$, that does not satisfy $U$, via the following formula $\mu_{n,\textrm{neg}}$:
\[U(q_0) \rightarrow \neg U(w).\]

Then, we introduce the relation $F$ via the formula $\mu_{n,\textrm{unif.}}$:
\begin{flalign*}
	\big[&\neg U(q_0)\wedge \neg U(q_1)\big] \rightarrow\\
	&\big[F(q_0, \ldots, q_0\mid q_0\mid q_0, w\mid q_0, \ldots, q_0)\wedge U(w)\big].
\end{flalign*}

The formulas $\mu_{n,\textrm{neg}}$ and $\mu_{n,\textrm{unif.}}$ together ensure the existence of two elements $q_0^\fA$ and $q_1^\fA$, the former not satisfying $U$, the latter satisfying it, such that the atom $F(q_0^\fA, \ldots, q_0^\fA \mid q_0^\fA \mid q_0^\fA, q_1^\fA\mid q_0^\fA, \ldots, q_0^\fA)$ is present.

It remains to replace the first arguments by new elements that would satisfy $U$.

This is done in a third step, where we shift to the right the first $n$ arguments, by introducing a new element that does satisfy $U$. We do it via the formula $\mu_{n, \textrm{shift}}$:
\begin{flalign*}
	\big[&F(x_1,\ldots,x_{n}\mid y\mid q_0, q_1\mid r_1, \ldots, r_{n})\wedge \neg U(x_n)\big] \rightarrow\\
	&\big[F(w, x_1,\ldots, x_{n-1}\mid y\mid q_0, q_1\mid r_1, \ldots, r_{n})\wedge U(w)\big].
\end{flalign*}

After we apply $n$ times the formula $\mu_{n, \textrm{shift}}$, we obtain the satisfied atom $F(c_1^\fA ,\ldots, c_n^\fA \mid q_0^\fA \mid q_0^\fA , q_1^\fA \mid q_0^\fA ,\ldots, q_0^\fA)$, where each $c_i^\fA$ satisfies~$U$.

This is exactly the tuple we wanted to obtain. We can now transfer it to the relation $P$, which is now of arity $2n{+}3$, via the formula $\mu_{n,\textrm{first perm}}$:
\begin{flalign*}
	\big[&F(x_1,\ldots,x_{n}\mid y\mid q_0, q_1\mid r_1, \ldots, r_{n})\wedge U(x_n)\big]\rightarrow\\
	&P(x_1,\ldots, x_{n}\mid y\mid q_0, q_1\mid r_1,\ldots, r_{n}).
\end{flalign*}

All of these formulas ensure together that the atom $P(c_1^\fA,\ldots, c_{n}^\fA \mid q_0^\fA \mid q_0^\fA , q_1^\fA \mid q_0^\fA ,\ldots, q_0^\fA)$ holds. This was the base of our original formula $\varphi_n$.

Then, the rest of the formulas are very similar to the original ones: we keep the same conditions written with the same relations $P$, $W$, $C_\rightarrow$, $S$, $C_{\leftarrow}$, and $Z$. The only difference is the arity of these symbols, changed from $2n{+}1$ to $2{n}+3$, in order to keep track of the variables $q_0$ and $q_1$.

Let us illustrate it with, for instance, the implication
\begin{flalign*}
W(x_1,\ldots,x_{n}\mid y\mid r_1, \ldots, r_{n-2}, q_0, q_0) \rightarrow\\
C_\rightarrow(x_2,\ldots, x_n, x_1\mid y\mid q_1, r_1, \ldots, r_{n-2}, q_0)
\end{flalign*}

from the original formula $\mu_{n,\textrm{cyclic}}$. Now it becomes the barely changed:
\begin{flalign*}
	&W(x_1,\ldots,x_{n}\mid y\mid q_0, q_1\mid r_1, \ldots, r_{n-2}, q_0, q_0)\rightarrow\\
	&C_\rightarrow(x_2,\ldots, x_n, x_1\mid y\mid q_0, q_1\mid q_1, r_1, \ldots, r_{n-2}, q_0).
\end{flalign*}
Another example from the original formula $\mu_{n,\textrm{cylic}^{-1}}$:
\begin{flalign*}
  &C_\leftarrow(x_1,\ldots,x_{n}\mid y\mid q_1, q_1, r_1, \ldots, r_{n-2}) \rightarrow\\
  &C_\leftarrow(x_n, x_1,\ldots, x_{n-1}\mid y\mid q_1, r_1, \ldots, r_{n-2}, q_0)
\end{flalign*}
is turned into
\begin{flalign*}
  &C_\leftarrow(x_1,\ldots,x_{n}\mid y \mid q_0, q_1 \mid q_1, q_1, r_1, \ldots, r_{n-2}) \rightarrow\\
  &C_\leftarrow(x_n, x_1,\ldots, x_{n-1} \mid y \mid q_0, q_1\mid q_1, r_1, \ldots, r_{n-2}, q_0).
\end{flalign*}

Basically, every implication in $\varphi_n$ is changed in a similar way. As previously stated, the final formula is obtained as $\forall x_1,\ldots, x_n,y, q_0, q_1,r_1, \ldots, r_n.~\exists w.~\psi_n$, where $\psi_n$ is the conjunction of the four new formulas $\mu_{n,\textrm{neg}}$, $\mu_{n,\textrm{unif.}}$, $\mu_{n, \textrm{shift}}$, and $\mu_{n,\textrm{first perm}}$, and of the updated versions of the original $\mu$'s.

\section{Proof of Proposition~\ref{proposition:param-dmk-model-size}}
\label{appendix:skolem}

This appendix is dedicated to the proof of Proposition~\ref{proposition:param-dmk-model-size}.

However, before proceeding to the proof, let us remark that the small model construction from Section~\ref{section:models} indeed requires an augmentation, as by the lower bound stated in Lemma~\ref{lemma:colourful-tournaments}, it cannot produce models smaller than $2^{\Omega(|\phi| \cdot\log |\phi|)}$, even in the parametrised setting (notice that here we do not restrict the use of existential quantifiers, hence the set $\cQ$ of arc colours can have size linear in $|\phi|$).

In the remaining of this appendix, we first give some intuitions and the full proof of the simpler case of the $\ParamSkolem{k}$ class,
and then we discuss necessary changes needed to obtain the proof for the class $\ParamDMK{k}$.
It is because we believe that the formal proof of the latter is technically more involved due to alternating quantifiers.

In the following proposition, $\ParamSkolem{k}$ stands for the formulas that are both in $\ParamDMK{k}$ and in $\Skolem$.

\begin{restatable}{proposition}{skolem}
	\label{proposition:skolem-model-size}
    Let $k$ be a natural number. If a sentence $\phi$ in $\ParamSkolem{k}$ is satisfiable, then it has a finite model of size $2^{\cO(|\phi|)}$.
\end{restatable}

In the following, we fix a sentence $\phi$ in $\ParamSkolem{k}$, which we assume to be satisfiable, and we show that it has a model of size $2^{\cO(|\phi|)}$.
Additionally, we assume that this sentence does not belong to any class $\ParamSkolem{k'}$ with $k'{<}k$, i.e. it has precisely $k$ universal quantifiers.

Denote by $\bK$ the grade of $\phi$.
Then $\phi$ has the shape $\forall \xs.~\forall \ys.~\exists \zs.~\psi$,
where $\xs$ are the special variables (at the number of $\bK\leq k$), $\ys$ are the rest of the universally quantified variables (at the number of $L = k{-}\bK$), $\zs$ are the existentially quantified variables (say at the number of $M$ and assume that $M>0$), and the formula $\psi$ is quantifier-free.

We shall briefly recall the main ideas of the construction of a model $\str{A}$ of size $2^{\cO(|\varphi|\cdot\log|\varphi|)}$ for $\varphi$.
By Lemma~\ref{l:lesstypes}, we know that Eloisa admits a winning strategy $\omegas$ in the game $\SAT(\varphi,\BBB)$, where $\BBB$ is a consistent and closed set of outer-types with exponentially many $1$-types.
The unnamed elements of the model $\str{A}$ are the vertices of an $(\cR,\cQ)$-paradoxical colourful tournament, where
the set $\cQ$ of arc colours is defined to be the set of variables from $\varphi$, and the set $\cR$ of vertex colours is defined to be the set (quotiented by some equivalence relation $\sim$) of reachable positions in the game $\SAT(\varphi,\BBB)$ when following the strategy $\omegas$.
These colours play a crucial role in the definition of the interpretations of the relation symbols of $\sigma(\varphi)$ in $\str{A}$.

This construction of $\str{A}$ goes in parallel with the construction of a winning strategy for Eloisa in the game $\VER(\varphi,\str{A})$, in order to prove that $\str{A}$ is indeed a model of $\varphi$.

Yet, in this case where $\varphi$ is in $\ParamSkolem{k}$, we can remark that $\VER(\varphi,\str{A})$ has a special shape, as it consists of two clearly distinct steps: first, Abelard chooses $k$ elements for the universally quantified variables $\xs$ and $\ys$, and then Eloisa chooses $M$ elements for the existentially quantified variables~$\zs$.

This dynamics of the games are therefore very much changed, as once Eloisa is on turn, she does not have to worry about Abelard's choices anymore. This justifies the notion of \emph{witness chains}, that basically stands for the choice of the $M$ elements at the same time.
Consider $\cT=(V,E)$ an $(\cR,\cQ)$-colourful tournament (not necessarily $(\cR,\cQ)$-paradoxical),
and let $\bar{\rho}$ be a play in $\SAT(\phi,\BBB)$, that is a sequence of consecutive positions $\rho_0,\rho_1,\dots,\rho_{L+M}$ that can be obtained by following Eloisa's winning strategy $\omegas$: the positions $\rho_0,\ldots, \rho_L$ are chosen by Abelard, while the positions $\rho_{L+1},\ldots, \rho_{L+M}$ are chosen by Eloisa.
We say that the $M$-tuple $\bar{b}=\tuple{b_1,\dots,b_M}$ of unnamed elements in $V$ forms a \emph{witness chain} for $\bar{\rho}$ if:
\begin{itemize}[nosep]
	\item for each $i\in [M]$, we have that $\rho_{L{+}i} \sim \mu(b_i)$;
	\item for all $i<j$, we have an arc $b_i \leftarrow b_j$, and moreover its colour $\lambda(b_i \leftarrow b_{j})$ is the variable $z_i$.
\end{itemize}

Moreover, let us consider $f\colon \xs\cup\ys \to V^\Cons$ an assignment. Let $\bar{a}=\tuple{a_1,\dots,a_\ell}$ be an enumeration of $f(\xs\cup\ys)\setminus\Cons$, for some $\ell\leq k$, and let $\bar{v}=\tuple{v_1,\dots,v_\ell}$ be a tuple of variables such that $f(v_i)=a_i$ for each $i\in[\ell]$ (the choice for $\bar{v}$ might not be unique). Then we say that $\bar{b}$ \emph{dominates} the assignment $f$ if, for every $t\in[M]$, $b_t$ colourfully dominates $\bar{a}$ via (a vertex colour $\sim$-equivalent to) $\rho_{L+t}$ and $\bar{v}$.
 
By the definition of being a witness chain for $\bar{\rho}$, $b_t$ will also colourfully dominate the $(\ell{+}t{-}1)$-tuple $\langle a_1,\dots,a_\ell, b_1, \ldots, b_{t-1}\rangle$ via (a vertex colour equivalent to) $\rho_{L{+}t}$ and $\langle v_1,\dots,v_\ell, z_1, \ldots, z_{t-1}\rangle$. Hence, if a structure $\str{A}$ is defined on $V$ in a manner similar to what we did in Section~\ref{section:models}, and if Abelard decided to choose the assignment $f\colon\xs\cup\ys\to V^\Cons$ in the game $\VER(\phi,\str{A})$, then by selecting the tuple $\bar{b}$ for $\zs$, Eloisa would win.

This is why our goal is now to construct an $(\cR,\cQ)$-colourful tournament of size $2^{\cO(|\phi|)}$ such that for every possible play $\bar{\rho}$ obtainable when following $\omegas$, and every assignment $f\colon \xs\cup\ys \to V^\Cons$, there exists a witness chain $\bar{b}$ for $\bar{\rho}$ that dominates $f$. Indeed, a structure constructed from such a tournament would be a model of $\phi$, as Eloisa would have a winning strategy for the game $\VER(\phi,\str{A})$.

\medskip

This is what we do now. Our construction of a model of size $2^{\cO(|\varphi|)}$ for $\varphi$ is outlined in three steps as follows: \begin{enumerate}[nosep]
	\item \label{step:witnesses}we show how to construct an $(\cR,\cQ)$-colourful tournament of size $2^{\cO(|\varphi|)}$ with witness chains;
    \item \label{step:structure}we turn this tournament into a $\sigma(\varphi)$-structure $\str{A}$ exactly as we did in Section~\ref{section:models};
	\item \label{step:model}finally, we show that this obtained structure is indeed a model of $\varphi$.
\end{enumerate}

\medskip
\noindent
\emph{Step~(\ref{step:witnesses}).}
We start our construction by obtaining a colourful paradoxical tournament. The colour sets we choose here for parameters are modified in comparison to the ones in Section~\ref{section:models} and recalled above: $\cQ'$ is now the set of variables $\xs \cup\ys$, of size $k$, and $\cR'$ is the set
$\cR \times [k{+}1] \times [M]$.
By Lemma~\ref{lemma:colourful-tournaments}, we know that there exists an $(\cR',\cQ')$-paradoxical colourful tournament $\big(\cT'{=}(V',E'),\mu',\lambda'\big)$ of size $2^{\cO(k\cdot\log k)} \times |\cR'|{\cdot}\log|\cR'|$. We have $|\cR'|=\cO(|\varphi|\cdot|\cR|)$, and recall from Section~\ref{section:models} that $|\cR| = 2^{\cO(|\phi|\cdot\log|\phi|)}$.

However, in our case, the size analysis of $\cF^\exists_{\omegas}$ given in Section~\ref{s:types}, from which the bound on $|\cR|$ follows, can be refined to $M\cdot(k+|\Cons|)^k = \cO(|\phi|^{k+1})$.
Hence, since the number $k$ is a constant in our problem, the tournament $\cT'$ is of size $2^{\cO(|\varphi|)}$ as desired.

The vertex set $V'$ of $\cT'$ can be arranged as a matrix with $k{+}1$ rows and $\bM$ columns, in accordance to the vertex colours, i.e. for $i\in[k{+}1]$ and $j\in[\bM]$, we define $V'_{i,j}$ to be the subset of vertices $\{ a \in V': \mu(a)=(r,i,j)\textrm{ for some $r\in\cR$}\}$. The point of this matrix is to be able to create $k{+}1$ disjoint groups of witness chains: one for each row $i$, the $j$'th element of any chain being in column $j$.

Now, we transform the $(\cR',\cQ')$-paradoxical colourful tournament $(\cT',\mu',\lambda')$ into a new $(\cR,\cQ)$-colourful tournament $\big(\cT{=}(V,E),\mu,\lambda\big)$ with witness chains as follows.

The vertex set $V=V'$ of $\cT$ is the same as that of $\cT'$, divided with the same cells $V_{i,j}=V'_{i,j}$.
The labelling $\mu$ is induced from the first coordinate of $\mu'$: for every $a\in  V$, $\mu(a)=r$ if $\mu'(a) = (r,i,j)$ for some $i,j$.
Now, let $a\in V_{i,j}$ and $b\in V_{i',j'}$ be two distinct vertices.
We put an arc $a \rightarrow_E b$ if one of the two conditions hold: either $a$ and $b$ are on the same row, i.e. $i = i'$, and $j > j'$,
in which case the label $\lambda(a\arc_E b)$ is defined as the variable $z_{j'}$;
or $a$ and $b$ are on different rows, i.e. $i \neq i'$, and there is an arc $a \arc_{E'} b$ in $\cT'$, in which case $\lambda (a \arc_E b)$ is $\lambda(a\arc_{E'}b)$, meaning that the arc and its label is defined accordingly to $\cT'$.

Refer to Figure~\ref{fig:tournament} for a depiction of the global shape of the so-defined tournament $\cT$, in the case where $k=3$ and $M=5$ (i.e. the prefix of $\varphi$ is for instance ``$\forall x_1, x_2.~\forall y.~\exists z_1, z_2, z_3, z_4, z_5$''). Naturally, not all the arcs are depicted.

\begin{figure}[h!]
	\begin{center}
		\includegraphics[scale=0.75]{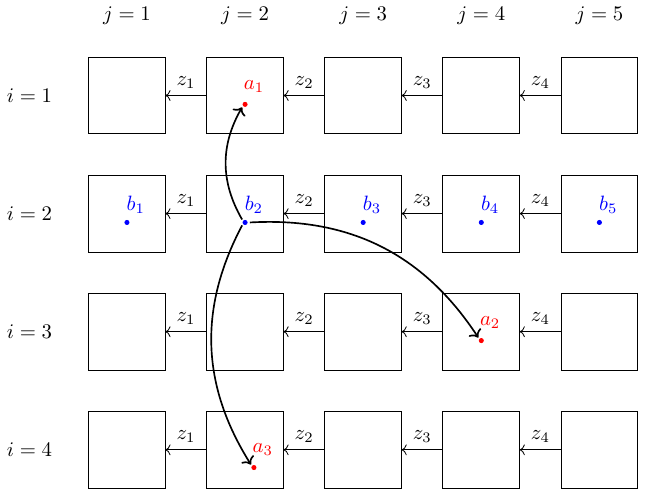}
        \caption{Eloisa's responses to Abelard's choices in $\cT$.}
		\label{fig:tournament}
	\end{center}
\end{figure}

\medskip
\noindent
\emph{Step~(\ref{step:structure}).} The so-defined $(\cR,\cQ)$-colourful tournament $(\cT,\mu,\lambda)$ is not an $(\cR,\cQ)$-paradoxical one
(to see this, e.g. take any tuple $\bar{a}$ from the last column, it cannot be colourfully dominated via any vertex colour $r \in \cR$ and any tuple of the variables in $\zs$, as there are no arcs with colours from $\zs$ entering $\bar{a}$).
Nevertheless, we still apply the same steps of the small model construction as described in Section~\ref{section:models}.
We briefly recall the three steps, without going in all the details:

We take the vertex set $V$ of $\cT$ as the unnamed domain. First, we assign the 1-types based solely on vertex colours.
Next, we define the hull-types on self-dominating subsets $S$, accordingly to the 1-types of elements in $S \setminus \{a\}$ and the vertex colour of $a$, where $a$ dominates $S\setminus\{a\}$.
Finally, we complete the structure by defining remaining hull-types on subsets of size between $2$ and $k{+}M$, we do it by selecting an outer-type from $\BBB$ in concordance with the $1$-types of the different elements in the subset.

\medskip
\noindent
\emph{Step~(\ref{step:model}).}
Now, we prove that our so-obtained structure $\str{A}$ is indeed a model of $\varphi$, which will conclude the proof of Proposition~\ref{proposition:skolem-model-size}.

In the proof below, we are not going to directly construct a simulation between the games $\VER(\phi,\str{A})$ and $\SAT(\phi,\BBB)$.
Instead, we will refer to the property of the construction from Section~\ref{section:models},
which states that it is enough for Eloisa to choose elements colourfully dominating the previously picked ones.

\begin{claim}
  Eloisa has a winning strategy in $\VER(\phi,\str{A})$.
\end{claim}
\begin{proof}
  Suppose that Abelard has constructed an assignment $f\colon \xs\cup \ys \rightarrow A^\Cons$ in Rounds $0,\dots,L$.
  Let $\bar{a}=\tuple{a_1,\dots,a_\ell}$ be an enumeration of elements in $f(\xs\cup \ys) \setminus \Cons$, for some $\ell\leq k$,
  and let $\vs=\tuple{v_1,\dots,v_\ell}$ be a tuple of variables such that $f(v_i) = a_i$ for each $i \in [\ell]$ (the choice of $v_i$ might not be unique).

  Now, we define Eloisa's response.
  Let $\bar{\rho}$ be the corresponding partial play in the game $\SAT(\phi,\BBB)$ simulated in parallel.
  As noticed before, due to special dynamics of the considered game, it can be uniquely extended to the full play in which Eloisa is a winner.

  Since the unnamed domain of $\str{A}$ is partitioned into $k{+}1$ rows,
  there exists $i \in [k{+}1]$ such that the $i$'th row is disjoint from $\bar{a}$, i.e. $\bar{a}\cap V_{i,j} = \emptyset$ for every $j\in[M]$.

  In Round $L{+}j$, with $j\in[M]$, let $r_j \in \cR$ be a vertex colour such that $r_j \sim \rho_{L+j}$.
  Then Eloisa selects for the variable $z_j$ an element $b_j \in V_{i,j}$ that colourfully dominates $\bar{a}$ via $r_j$ and~$\vs$.

  We argue that such a choice of $b_j$ is possible:
  consider the tuple $\bar{a}$ as the vertices of the $(\cR',\cQ')$-paradoxical colourful tournament $(\cT',\mu',\lambda')$.
  Thus, we can select a vertex $b_j$ dominating $\bar{a}$ via $(r_j,i,j)$ and $\bar{v}$.
  This element $b_j$ lies in $V_{i,j}$, and therefore, in $\cT$, the arcs between $b_j$ and the tuple $\bar{a}$ are exactly as in $\cT'$, and their labellings as well.
  As the consequence, in $\cT$, $b_j$ colourfully dominates $\bar{a}$ via $r_j$ and $\bar{v}$, as desired.

  Moreover, all the witnesses selected by Eloisa come from the same $i$'th row and from the consecutive columns, i.e. each witness $b_j$ lies in the $j$'th column.
  Thus, by the construction of $(\cT,\lambda,\mu)$, the tuple $\tuple{b_1,\dots,b_\bM}$ forms a witness chain dominating $f$.
  
  Figure~\ref{fig:tournament} depicts Eloisa's response in the case where $\varphi$ is of the shape $\forall x_1, x_2.~\forall y.~\exists z_1, \ldots, z_5.~\psi$.  Abelard's assignment $f$ maps $x_1$ to $a_1$, $x_2$ to $a_2$, and $y$ to $a_3$, in rows $1$, $3$, and $4$ respectively. As he did not select any element from the second row, Eloisa is happy to respond there.

  Therefore, we conclude that the conditions of Claim~\ref{c:correspondance2} hold in the same way here,
  as the $1$-types and hull-types were defined accordingly as in Section~\ref{section:models}, so Eloisa wins the game $\VER(\varphi,\str{A})$.
\end{proof}

We remark that this construction stays deterministic if, for $\cT'$, we employ the explicit construction of paradoxical colourful tournaments from Appendix~\ref{appendix:graphs}.

\medskip
Now, we explain how the idea of witness chains can be generalised in order to obtain the same bound for the larger class $\DMK^{\forall=k}$.

\paramdmk*

For the rest of this appendix, we fix a satisfiable sentence $\phi$ in $\DMK$ with $k$ universal quantifiers of the shape as in (\ref{eq:prenex}), that is $\forall x_1 \ldots \forall x_\bK.~\Qfr_{1}y_{1} \ldots \Qfr_\bM y_\bM.~\psi$.
We denote by $L$ the number of existentially quantified variables in $\phi$, i.e. $\bM{+}\bK{-}k$, and let $y_{j_1},\dots,y_{j_L}$ be the existentially quantified variables (in this order).

The strategy to prove Proposition~\ref{proposition:param-dmk-model-size} is in essence really similar.
The main difference is that, instead of a grid $[k{+}1]\times[M]$ as the range of second and third coordinates for $\cR'$, we choose the set of nodes of the $(k{+}1)$-branching tree $\treet$ with $L{+}1$ levels (the root has level $0$ and plays only an auxiliary role).
Such a tree has size $k^{\cO(L)}$, which is still $2^{\cO(|\phi|)}$ in our case.
Formally, $\cR'$ is now $\cR{\times}(\treet{\setminus}\{ \textrm{root} \})$.
On the contrary, the set $\cQ'$ does not need any adjustments, that is, it still consists of universally quantified variables, and hence has a fixed size $k$.
We obtain an $(\cR',\cQ')$-paradoxical colourful tournament $(\cT'=(V',E'),\mu',\lambda')$ of size $2^{\cO(|\phi|)}$.

As in the case of $\ParamSkolem{k}$, we transform $\cT'$ into an $(\cR,\cQ)$-colourful tournament $(\cT=(V,E),\mu,\lambda)$.
Our goal is to obtain witness chains along the paths going from the children of the root (level $1$) to the leaves (level $L$).
We do this by modifying the arcs connecting pairs of vertices whose second coordinates correspond to the nodes being in the ancestor-descendant relation in the tree.
Below we give the formal details of this augmentation.

The vertex set $V=V'$ of $\cT$ is the same as that of $\cT'$.
We partition the vertex set $V$ into groups $V_u$, for each $u\in\treet \setminus \{ \textrm{root} \}$, corresponding to the same node of the tree $\treet$, i.e. $V_u$ is the set $\{ a \in V : a = (r,u) \text{ for some } r \in \cR \}$.
The labelling $\mu$ is induced from the first coordinate of $\mu'$: for every $a\in  V$, $\mu(a)=r$ if $\mu'(a) = (r,v)$ for some $v\in\treet$.
Now, let $a \in V_u$ and $b \in V_{u'}$ be two distinct vertices.
We put an arc $a \rightarrow_E b$ if one of the two conditions hold: either $u$ is a (not necessary immediate) strict descendant of $u'$ in $\treet$, in which case the label $\lambda(a\arc_E b)$ is defined as the variable $y_{j_\ell}$, where $\ell$ is the level of $u'$;
or $u$ and $u'$ do not lie on the same branch of $\treet$, and there is an arc $a \arc_{E'} b$ in $\cT'$, in which case $\lambda (a \arc_E b)$ is $\lambda(a\arc_{E'}b)$, meaning that the arc and its label is defined accordingly to $\cT'$.

We obtain a final structure $\str{A}$ by again employing the model construction described in Section~\ref{section:models} to the tournament $\cT$.
The reader can notice that, by the definition above, indeed the paths going downwards the tree create desired witness chains.
We finish by sketching Eloisa's winning strategy in $\VER(\phi,\str{A})$.

The high level idea for her strategy is that she dominates the previous choices of Abelard in a branch where he did not select any elements.
Since he chooses at most $k$ elements during the entire play, there is necessarily a branch he does not visit, hence she can continue with a witness chain existing there, even if Abelard interrupted her due to alternating quantifiers. Naturally, in her move corresponding to the variable $y_{j_i}$, she selects an element from the $i$'th level of the tree.
This way, she can maintain the invariants stated in Claim~\ref{c:correspondance2}, and hence she wins the game.

\section{Expanding the universal-uniform fragment to $\overline{\text{DK}}$}
\label{appendix:application}

In this appendix, we provide the expansion from $\FAUF$ to \DMDK, or, actually, to conjuctions of \Skolem-sentences, mentioned in Section~\ref{section:application}:

\expandingfauf*

\begin{proof}
	Let $\phi$ be a sentence in \FAUF{}. For every subformula $\mu$ of $\phi$ that is of the shape $\forall x,\ys.~\nu(x,\ys)$ (resp. $\forall \ys.~\nu(x,\ys)$), we introduce a fresh relational symbol $P_\mu$, of arity $0$ (resp.~$1$) (i.e. the arity is the number of free variables of $\mu$). The signature $\sigma(\psi)$ will consist of $\sigma(\phi)$ expanded of all these $P_\mu$'s.
	
	Now, we consider a transformation $\tr[\mu]$ for every subformula $\mu$ of $\phi$, defined inductively:
	\begin{itemize}[nosep]
		\item the transformation of every literal is itself;
		\item the transformation of every subformula of the shape $\mu\oplus\nu$, with $\oplus\in\{\wedge,\vee\}$ and possible free variables, is $\tr[\mu]\oplus\tr[\nu]$;
		\item the transformation of every subformula of the shape $\exists\ys.~\nu(\xs,\ys)$ is $\exists\ys.~\tr[\nu(\xs,\ys)]$;
		\item the transformation of every subformula $\mu$ of the shape $\forall x,\ys.~\nu(x,\ys)$, with no free variables, is the atom $P_\mu$;
		\item the transformation of every subformula $\mu(x)$ of the shape $\forall\ys.~\nu(x,\ys)$, with one free variable $x$, is the atom $P_\mu(x)$.
	\end{itemize}
	
	It is immediate to see that for any subformula $\mu$ of $\phi$, the obtained $\tr[\mu]$ is in the fragment $\exists^\ast$ (when converted to prenex form).
	
	Then, for every subformula $\mu$ of $\phi$ starting with an universal quantifier, we define a formula $\ax[\mu]$ axiomatising the relational symbol $P_\mu$, as follows:
	
	\begin{itemize}
		\item if $\mu$ is $\forall x,\ys.~\nu(x,\ys)$, with no free variable, then $\ax[\mu]$ is the formula \[\forall x,\ys.~P_{\mu}\rightarrow\tr[\nu(x,\ys)];\]
		\item if $\mu$ is $\forall\ys.~\nu(x,\ys)$, with one free variable $x$, then $\ax[\mu]$ is the formula \[\forall x,\ys.~P_{\mu}(x)\rightarrow\tr[\nu(x,\ys)].\]
	\end{itemize}
	
    It is easy to check that $\ax[\mu]$ is in $\DMK$. Indeed, in both shapes above, if we consider the set $\vs\subseteq\{x\}\cup\ys$ from the definition of $\FAUF{}$, then by the same definition, every literal of $\nu(x,\ys)$ not bounded by quantifiers admits exactly $\vs$ as its set of variables. Since $\tr$ does not change literals, the same property holds for $\tr[\nu(x,\ys)]$, and $\vs$ can therefore be taken as the special variables. Moreover, since each $\tr[\nu(x,\ys)]$ is in $\exists^\ast$ (in prenex form), $\ax[\mu]$ is indeed in $\Skolem$ (in prenex form).
	
	Finally, we define the desired formula $\psi$ as the conjunction of $\tr[\phi]$ and of the $\ax[\mu]$'s, with $\mu$ ranging over the subformulas of $\phi$ starting with a universal quantifier.
	
	It is readily verified that any model $\str{A}$ of $\phi$ can be expanded to a model of $\psi$, with the interpretation of any $P_\mu$ being the set $\{a\in A\mid \str{A}\models\forall \ys.~\nu(a,\ys)\}$ when $\mu$ is as above and with one free variable, and true/false when it has no free variables (true/false depending whether $\str{A}\models\forall x,\ys.~\mu(x,\ys)$ or not).
	
	Reciprocally, a proof that any model of $\psi$ is a model of $\phi$ as well goes by a simple induction over $\phi$.
\end{proof} 
\end{document}